\documentclass[a4paper,11pt]{article}
\usepackage{fullpage}

\usepackage[english]{babel}
\usepackage[T1]{fontenc}
\usepackage{hyphenat}

\usepackage{amsthm, amsmath, amssymb}
\usepackage{mathtools}
\usepackage{enumitem}
\usepackage{nicefrac}
\usepackage{float}
\usepackage{color}
\usepackage{bbm}
\usepackage[margin=1.5cm, font=small]{caption}
\usepackage{subcaption}
\usepackage{algorithm}
\usepackage[noend]{algpseudocode}
\let\Algorithm\algorithm
\renewcommand\algorithm[1][]{\Algorithm[#1]\setstretch{1.15}\small}
\newcounter{procedure}
\makeatletter
\newenvironment{procedure}[1][htb]{%
  \let\c@algorithm\c@procedure
  \renewcommand{\ALG@name}{\textbf{Procedure}}
  \begin{algorithm}[#1]%
  }{\end{algorithm}
}
\makeatother
\usepackage{makecell}
\usepackage{booktabs}
\usepackage{varwidth}
\usepackage{array}
\usepackage{float}
\usepackage{csquotes}
\usepackage[sort&compress,numbers]{natbib}
\makeatletter
\def\NAT@def@citea{\def\@citea{\NAT@separator}}%
\makeatother
\usepackage{comment}
\usepackage{graphicx}
\usepackage{cellspace}
\usepackage{setspace}
\usepackage[hypertexnames=false]{hyperref}
\hypersetup{
    colorlinks=true,             
    linkcolor=[rgb]{0, 0.3, 0.6},          
    citecolor=[rgb]{0.5, 0.5, 0.5}, 
    urlcolor=black,               
    linktocpage=true
}
\usepackage[nameinlink]{cleveref}

\theoremstyle{plain}
\newtheorem{theorem}{Theorem}[section]
\newtheorem{lemma}[theorem]{Lemma}
\newtheorem{corollary}[theorem]{Corollary}
\newtheorem{claim}[theorem]{Claim}

\theoremstyle{definition}
\newtheorem{remark}{Remark}

\crefname{claim}{claim}{claims}
\crefname{line}{line}{lines}
\Crefname{procedure}{Procedure}{Procedures}

\newcommand{\rep}{\backslash}
\newcommand{\eps}{\varepsilon}
\newcommand{\OPT}{\mathsf{OPT}}
\newcommand{\TP}{\mathsf{TP}}
\renewcommand{\k}{^{(\mu)}}
\newcommand{\X}{X^\star}
\newcommand{\V}{\overline{V}}
\newcommand{\ddelta}{\dot{\delta}}

\newcommand{\cmin}{c^{\min}}

\newcommand{\+}{$_{\hspace{-0.15em}+}$}
\newcommand{\0}{$_{0}$}
\newcommand{\1}{$_{0\hspace{0.065em}}$}

\renewcommand{\leq}{\leqslant}
\renewcommand{\geq}{\geqslant}

\DeclareMathOperator\argmin{argmin}
\DeclareMathOperator\cost{cost}
\DeclarePairedDelimiter{\abs}{\lvert}{\rvert}
\DeclareMathOperator\MSF{MSF}

\title{Efficient Approximations for Many-Visits\\ Multiple Traveling Salesman Problems}

\author{Krist{\'o}f B{\'e}rczi\thanks{MTA-ELTE Momentum Matroid Optimization Research Group and MTA-ELTE Egerváry Research Group, Department of Operations Research, E{\"o}tv{\"o}s Lor{\'a}nd University, Budapest, Hungary. \texttt{kristof.berczi@ttk.elte.hu}.} 
  \and Matthias Mnich\thanks{TU Hamburg, Hamburg, Germany. \texttt{matthias.mnich@tuhh.de}.}
  \and Roland Vincze\thanks{Universit\"{a}t Augsburg, Augsburg, Germany. \texttt{roland.vincze@uni-a.de}.}
}

\begin{document}
\date{}
\maketitle

\begin{abstract}
A fundamental variant of the classical traveling salesman problem (TSP) is the so-called \emph{multiple TSP} (mTSP), where a set of $m$ salesmen jointly visit all cities from a set of $n$ cities.
The mTSP models many important real-life applications, in particular for vehicle routing problems.
An extensive survey by Bektas (\emph{Omega} \textbf{34}(3), 2006) lists a variety of heuristic and exact solution procedures for the mTSP, which quickly solve particular problem instances.

In this work we consider a further generalization of mTSP, the \emph{many-visits mTSP}, where each city $v$ has a request $r(v)$ of how many times it should be visited by the salesmen.
This problem opens up new real-life applications such as aircraft sequencing, while at the same time it poses several computational challenges.
We provide multiple efficient approximation algorithms for important variants of the many-visits mTSP, which are \emph{guaranteed} to quickly compute high-quality solutions for \emph{all} problem instances.
\end{abstract}

\section{Introduction}
\label{sec:introduction}
The traveling salesman problem (TSP) is one of the cornerstones of combinatorial optimization.
In the classical TSP, a set $V$ of $n$ cities is given, along with their pairwise distances $c(uv)$, and one seeks a tour of minimum cost that visits every city exactly once.
A well-known extension of the TSP is the \emph{multiple TSP} (mTSP), where there is not just one but a set of $m$ salesmen (or \emph{agents}), and one seeks a solution such that each of the $n$ cities is visited by at least one of the salesmen.
There are different quality measures that can be associated with such tours, such as the overall length of all the tours or the maximum length of any of the tours of the salesmen.
Various models exist, for example based on whether the salesmen are required to start their tours at a fixed set $D$ of depots or not.

An excellent overview of those variations is provided by Bektas~\cite{Bektas2006}, who in their survey also describes exact and heuristic solutions for the mTSP.
Here, exact solution refers to a solver (implemented algorithm) which is guaranteed to return an optimal solution for a particular set of instances, but its worst-case running time might be exponential in the input size.
In turn, a heuristic is guaranteed to return a solution quickly, but does not come with guarantees on the quality of the computed solution.
It is possible to design algorithms for the mTSP that exhibit both advantages at the same time, namely fast run times and solution quality guarantees. These are so-called \emph{$\alpha$-approximation algorithms}, which return a solution whose cost is at most a factor $\alpha$ worse than that of an optimal solution, while their run time is guaranteed to be efficient (polynomial) for all problem instances.
For mTSP, such approximation algorithms are based on a reduction to the standard TSP, and then resorting to approximation algorithms for the standard TSP.
The standard TSP, with arbitrary cost function~$c$, does not admit $\alpha$-approximation algorithms for \emph{any} constant factor $\alpha$, assuming that $\mathsf{P\not=NP}$.
A widely studied restriction is the TSP with \emph{metric} cost function, which means that city distances are non-negative and satisfy the triangle inequality, an assumption that is justified by many practical applications.
For metric TSP, the best-known approximation factor obtainable in polynomial time was $\alpha = \nicefrac32$ for a long time, independently discovered by Christofides~\cite{Christofides1976} and Serdyukov~\cite{Serdyukov1978} in the 1970s.
Very recently, Karlin et al.~\cite{KarlinKG2021} improved the factor to $\alpha = \nicefrac32 - \varepsilon_0$ for some small constant $\varepsilon_0 \approx 10^{-36}$.

In this work we consider a wide generalization of the multiple TSP, which we call the \emph{many-visits multiple TSP} (many-visits mTSP).
In this problem, each city $v\in V$ is equipped with an integer request $r(v)\geq 1$, which specifies the total number of times that it should be visited by the salesmen.
Thus, on the one hand, the mTSP is the special case when all requests~$r(v)$ are equal to one.
On the other hand, the many-visits mTSP is also an extension of the so-called \emph{many-visits TSP}.
The many-visits TSP was introduced in 1966 by Rothkopf~\cite{Rothkopf1966} as a generalization of the TSP, where in addition to the pairwise city distances, a request $r(v)$ is given for each city $v$.
The goal is then to find a minimum-cost closed walk that visits each city~$v$ exactly~$r(v)$ times.
In its path variant, called the \emph{many-visits path TSP}, there is a designated starting point $s\in V$ and a designated end point $t\in V$, where possibly $s\not=t$, and the salesman has to find a walk of minimum cost starting at $s$, ending at $t$, and visiting each city $v$ exactly~$r(v)$ times.
B{\'e}rczi et al.~\cite{BercziMV2020} recently gave a polynomial-time $\nicefrac32$--approximation for the many-visits path TSP and many-visits TSP with metric cost functions.
We remark that the requests $r(v)$ for the cities $v$ can be \emph{exponentially large} in the number~$n$ of cities; it is therefore not possible to solve instances of the many-visits TSP by reducing them to the standard TSP and solving them with one of the available solution procedures for the standard TSP, as such reductions would lead to an exponential blow-up.

The many-visits TSP and the many-visits path TSP find important real-life applications in modeling high-multiplicity scheduling problems.
Suppose there are $r$ jobs of $n$ different types to be executed on a single, universal machine.
Processing a job of type $j$ bears a cost $p(j)$, while switching from a job of type $i$ to type $j$ costs~$s(ij)$, and the goal is to find a sequence of jobs with minimum total cost.
Modeling this problem as an instance of many-visits TSP is straightforward, by defining the job types as cities, the number of jobs as visit requests for a city, and finally letting $c(ij)$ denote the sum of~$p(j)$ and $s(ij)$.
Note that $c(ii)$ is not necessarily smaller than~$c(ij)$ for $i \neq j$, meaning that in some applications it is beneficial to switch between job types.
Such a problem belongs to the class of high-multiplicity scheduling problems with setup times~\cite{AllahverdiNCK2008}.
Recently, Jansen et al.~\cite{JansenKMR2019} and Deppert and Jansen~\cite{DeppertJansen2019} considered scheduling~$n$ jobs divided into $c$ classes on $m$ identical parallel machines, where the jobs have sequence-independent batch setup times; the authors provided linear-time approximation algorithms for $\alpha = \nicefrac32$ and polynomial-time approximation algorithms for $\alpha = 1+\varepsilon$ for any $\varepsilon > 0$, for several variants of this problem.

An economically highly important task is that of sequencing landing aircraft at airports, where often a single runway or at best a few runways form a capacity bottleneck to the landing operations.
This application was (to the best of our knowledge) first described by Psaraftis~\cite{Psaraftis1980} in 1980 for the single-runway scenario, when all flights arrive at a single point of time and the objective is to land the last aircraft as early as possible; they devised algorithms with exponential run time dependence on the number of aircraft classes and polynomial run time dependence in the total number of aircraft.
Psaraftis' setting essentially corresponds to the many-visits path-TSP for a single agent ($m = 1$).

\pagebreak
Various other aircraft sequencing models, with more complex and more realistic settings, were developed over the past four decades; we refer to Briskorn and Stolletz~\cite{BriskornStolletz2014} for an overview.
With regards to this application, the many-visits mTSP for variable $m > 0$ corresponds to the aircraft sequencing problem with $m$ runways, where the cities correspond to aircraft types, intercity distances to minimum pairwise separation distances between aircraft types to avoid woke turbulence, and the requests $r(v)$ for city $v$ to the number of aircrafts of type $v$.
In particular, there is a non-zero distance $c(vv) > 0$ between landing two aircrafts of the same type.
Typical objectives that arise in practise are the aforementioned \emph{makespan} of landing the last aircraft as early as possible, or the \emph{sum of weighted completion times objective}, where each aircraft is assigned a weight (or priority) and one wishes to minimize the product of the weights with the landing times, summed over all aircrafts.

\subsection{Our contributions}
In this work we design efficient approximation algorithms for the many-visits multiple TSP.
As mentioned, even for the case of a single salesman ($m = 1$) the potentially exponentially large requests $r(v)$ rule out the existence of an simple reduction to the standard TSP and employing off-the-shelf methods for it.
Therefore, we design sophisticated algorithms that simultaneously handle those large requests as well as the partitioning task of deciding for each of the $m$ salesmen how many times they visit the city $v$, for each $v\in V$.
Consequently, our algorithms need to solve a complex partitioning task, as well as a sequencing task to find the right order of visits, and perform all of this in polynomial time.
Eventually, our algorithms are guaranteed to always return solutions of cost not much more than the optimum ($\alpha \leq 4$), for all problem instances; the precise approximation factor $\alpha$ depends on the variant of the problem that we consider.

However, giving the precise definition of the many-visits mTSP is not straightforward; one can consider different, but equally valid descriptions when introducing multiple salesmen to the many-visits TSP.
First, there are various choices for the objective function: one possibility is to minimize the \textit{total} cost of the tours, and another is to minimize the cost of the \textit{longest} tour.
These objectives translate to minimizing the total processing time ($\sum C_j$) and makespan ($C_{\max}$), respectively, using the terminology of scheduling theory.
Observe that in the case of the many-visits TSP, where there is only one agent, the two objectives are equivalent.
Second, one can also impose different constraints on the tours sought in the solution, such as requiring the tours of different agents to be disjoint or forbidding empty tours, i.e.\ agents not visiting any city.
Finally, one can consider a special city for each salesman called its \textit{depot}, and each tour in the solution must contain exactly one depot; alternatively there are no depots and each agent simply traverses a subset of the cities.

The four aspects mentioned above are independent from each other, therefore all $2^4=16$ possible generalizations can be considered.
Our focus in this work is on minimizing the total cost of the tours as an objective, and we start by examining these 8 problem variants and their relations to each other.
In particular, we show which ones are equivalent and which ones are different in terms of their optimal solution values.
Thereafter, we develop our main contributions, which are efficient approximation algorithms with small approximation factors $\alpha$ for all problem variants.

\subsection{Formal problem description}
The many-visits multiple TSP (many-visits mTSP or MV-mTSP) is defined as follows.
Given is a complete graph $G(V,E)$ with non-negative costs $c(uv)$ for every pair of vertices $u, v$ and a positive request~$r(v)$ for each vertex $v$ (encoded in binary).
Let us denote by $\cmin(v) := \min_{u \in V-v} c(uv)$ the cheapest edge adjacent to $v$.
Besides the inequality $c(uw) \leq c(uv) + c(vw)$ for every triplet $u, v, w$, the cost function being metric means that the cost of the self-loop $c(vv)$ at each vertex $v\in V$ is at most the cost of leaving city $v$ to any other city~$u$ and returning\footnote{
Note that, unlike in most results involving metric costs, we do not assume that the cost of a self-loop is~$0$.}, that is:
\begin{equation*}
  c(vv) \leq 2 \cdot \cmin(v) \qquad \text{for all } v \in V \enspace .
\end{equation*}
Moreover, an integer $m$ is given, that denotes the (maximum) number of agents. 

Let us now introduce the different variants of the problem.
Based on whether the agents can start from arbitrary cities or only from their dedicated depots, we distinguish the cases when
\begin{itemize}
  \item[a)]    \textit{there are no depots / there is a set of depots $D = \{d_1, \dots, d_k\}$ such that $D \subseteq V$;} \label{dist:a}
  \end{itemize}
in case we consider no depots, we call the problem \textit{unrestricted} MV-mTSP.

Depending on whether an agent can be left without work or not, the goal is to find a multigraph $X$ such that
\begin{itemize}
  \item[b)]    \textit{$X$ can be decomposed into at most / exactly $m$ non-empty many-visits TSP tours;} \label{dist:b}
\end{itemize}
the corresponding problems are denoted with a $0$ or $+$ subscript, respectively. Here \textit{non-empty} means containing at least one edge for the unrestricted problem, or a depot and at least one non-depot vertex for the problem with depots. 

Finally, the MVTSP tours we are seeking for can be
\begin{itemize}
  \item[c)]    \textit{vertex-disjoint / not necessarily vertex-disjoint.} \label{dist:c}
\end{itemize}
In the former case, we say that the problem is \textit{with disjoint tours}, while in the latter \textit{with arbitrary tours} (referring to the possibility but not a requirement of overlapping tours).
The setting of arbitrary tours allows that two or more agents \emph{jointly} satisfy the visit request $r(v)$ of a city, in the sense that the visits of the agents sum up to $r(v)$.

Let us denote the many-visits TSP tour of agent $i$ by $X_i$, for $i=1, \dots, m$.
\Cref{tab:Px} contains an overview of the problems considered. Note that distinctions b) and c) do not appear in the single-visit mTSP and hence are introduced in this paper.

\begin{table}[!ht] \centering
\begin{tabular}{lrcc}
  \toprule
 & & \textbf{arbitrary tours} & \textbf{disjoint tours} \\ \midrule
$\sum_i \cost(X_i)$ 
 & unrestricted MV-mTSP\+ & P1 & P2 \\ 
 & unrestricted MV-mTSP\1 & P3 & P4 \\
 & MV-mTSP\+ & P5 & P6 \\ 
 & MV-mTSP\1 & P7 & P8 \\
 \bottomrule
\end{tabular}
\caption{Overview of the problems considered in this paper.}
\label{tab:Px}
\end{table}

\paragraph{Variants without depots.}
In the unrestricted MV-mTSP variants, the goal is to obtain a multigraph consisting of $m$ closed walks (MVTSP tours), that has a minimum total cost and visits each vertex $v$ exactly $r(v)$ times.
When it comes to precise definitions of this problem, one can set the rules for the $m$ tours in different ways.
Some of these differences might seem subtle, but they are crucial in algorithm design.
Recall that in the unrestricted setting, an MVTSP tour is non-empty if it contains at least one edge.
We denote the different possible problems the following way:

\pagebreak
\begin{itemize}
  \item[(P1)]  unrestricted MV-mTSP\+ with arbitrary tours: the edge set of the resulting multigraph $X$ can be partitioned into $m$ non-empty, not necessarily vertex-disjoint MVTSP tours;
  \item[(P2)]  unrestricted MV-mTSP\+ with disjoint tours: the resulting multigraph $X$ has exactly~$m$ components, each of them being a non-empty MVTSP tour;
  \item[(P3)]  unrestricted MV-mTSP\0 with arbitrary tours: the edge set of the resulting multigraph~$X$ can be partitioned into at most $m$ not necessarily vertex-disjoint MVTSP tours;
  \item[(P4)]  unrestricted MV-mTSP\0 with disjoint tours: the edge set of the resulting multigraph~$X$ has at most~$m$ components, each of them being an MVTSP tour.
\end{itemize}
  
It is worth emphasizing that the cost of a self-loop at a vertex, i.e.\ going from vertex $v$ to itself, does not necessarily incur a zero cost.
Moreover, if an agent is assigned a single vertex~$v$ in a solution, an isolated point does not count as one visit, i.e.\ the tour of that agent has to consist of a self-loop at $v$ for each visit they ought to make at~$v$.

\paragraph{Variants with depots.}
In the MV-mTSP variants (which contain depots), a set $D \subseteq V$ of special vertices called depots is given, and the goal is to find a minimum cost multigraph consisting of $m$ closed walks such that each one contains exactly one depot $d$ from $D$.
Recall that in the restricted setting, an MVTSP tour is non-empty if it contains a depot and at least one non-depot.
The problem definitions follow a similar pattern as in the unrestricted case, but include the notion of depots:
\begin{itemize}
  \item[(P5)]  MV-mTSP\+ with arbitrary tours: the edge set of the resulting multigraph $X$ can be partitioned into $m$ non-empty, not necessarily vertex-disjoint MVTSP tours;
  \item[(P6)]  MV-mTSP\+ with disjoint tours: the edge set of the resulting multigraph $X$ has exactly~$m$ components, each of them being a non-empty MVTSP tour;
  \item[(P7)]  MV-mTSP\0 with arbitrary tours: the edge set of the resulting multigraph $X$ can be decomposed into at most $m$ not necessarily vertex-disjoint MVTSP tours;
  \item[(P8)]  MV-mTSP\0 with disjoint tours: the edge set of the resulting multigraph $X$ has at most~$m$ components, each of them being a MVTSP tour.
\end{itemize}

Let us note that the depots have no $r(\cdot)$ values assigned.
However, the problem definitions inherently provide a lower bound for the number of visits.
  In the case of non-empty tours, the number of visits has to be at least one, while in the case of empty tours this lower bound is zero.
  There is no implicit upper bound for the number of visits, however, there is always an optimal solution where every depot is visited at most once, as in case of more than one visits it is possible to do shortcuts and obtain a tour with smaller or equal cost; see \Cref{lem:depotonce}.
\medskip

\paragraph{Implications in scheduling theory.}
Let us rephrase the problems using scheduling theory terminology.  
In the unrestricted case, the disjoint tours variant can be translated as finding a schedule of jobs such that each job of type $j$ is assigned to the same machine.
The arbitrary tours variant corresponds to finding a schedule of jobs, where jobs of the same type can be processed on multiple machines. 
Depending on whether empty tours are allowed or not, machines are allowed to be idle (i.e. with no jobs scheduled on it) or not, respectively.

The different variants with depots correspond to similar scheduling problems, where additionally a special job (that corresponds to the depot) is present on each machine.
One can think of this job as a preparation step that needs to be performed if one intends to process any jobs on that machine.

\subsection{Related work}
The generalization of the (single-visit) TSP to multiple salesmen or agents is commonly referred to as multiple TSP or mTSP.
The first constant-factor approximation algorithm for the mTSP is due to Frieze~\cite{Frieze1983}, who considered $m$ salesmen starting from a single vertex called \emph{depot}. 
They extended the Christofides-Serdyukov algorithm for the TSP~\cite{Christofides1976,Serdyukov1978}, and provided a $\nicefrac32$\hyp{}approximation.

A more diverse set of works considered a modification of the mTSP, where the salesmen start from different depots.
The multidepot mTSP (referred to as MDMTSP in the literature) seeks $m$ cycles that together cover a set of cities, and have exactly one depot in each cycle.
Rathinam et al.~\cite{RathinamSD2007} provided a $2$-approximation using a tree-doubling approach, then Xu et al.~\cite{XuXR2011} showed that a Christofides-like heuristic yields a $(2-\nicefrac1m)$-approximation.
Xu et al. first construct a minimum cost constrained spanning forest, where each tree contains exactly one depot, then complete it with a minimum cost matching on the odd degree vertices.
They also argue that the cost of the matching cannot be bounded by $\nicefrac{\OPT}{2}$, as in the case of TSP.
The reason is that the edges of the matching might go between vertices that are in different components of the optimal solution.
The $\nicefrac32$-approximation algorithm by Xu and Rodrigues~\cite{XuRodrigues2015} circumvents this problem by exchanging the edges of the constrained spanning forest, although this edge-exchanging procedure leads to a time complexity of $n^{\Omega(m)}$.
It is worth noting that the results so far allow isolated depots in the solution, i.e.\ a salesman that visits no cities.

In the generalized multidepot mTSP, more than $m$ depots are available, but only at most~$m$ can be selected.
The first constant-factor polynomial\hyp{}time algorithms are $2$-approximations by Malik et al.~\cite{MalikRD2007} and by Carnes and Schmoys~\cite{CarnesShmoys2011}.
Later, Xu and Rodrigues~\cite{XuRodrigues2017} provided a $(2-\nicefrac{1}{(2m)})$-approximation.
 The algorithm of Xu and Rodrigues~\cite{XuRodrigues2015} can be used to obtain a $\nicefrac32$-approximation for the generalized multidepot mTSP, albeit requiring~$n^{\Omega(m)}$ time.

Khachay and Neznakhina~\cite{KhachayNeznakhina2016} considered the relevant \emph{$m$-size cycle cover} problem, that seeks~$m$ cycles covering the set of cities with minimum total cost.
The authors show a $2$-approximation to this problem when the edge costs are asymmetric, satisfy the triangle inequality, and the cost of each loop is zero.

\subsection{Detailed results}
The main results of this paper include polynomial-time constant-factor approximation algorithms for all variants of the many-visits mTSP.

We first provide $3$- and $4$-approximations for the arbitrary tours variants (\Cref{alg:mamvtsp_3apx,alg:mdmvtsp_3apx}).
In the unrestricted setting, the algorithm computes a minimum cost constrained spanning forest, double its edges, and then adds self-loops to fix the degrees.
In the restricted setting, an optimal transportation problem solution is calculated instead of self-loops.
By bounding the cost of both of these structures, the approximation ratios follow.

We continue by showing $4$-approximations for the disjoint tours variants (\Cref{alg:mamvtsp_4apx,alg:mdmvtsp_4apx}). 
The proofs are similar to that of \Cref{alg:mamvtsp_3apx}; however, due to the slight differences caused by the number of agents we discuss those separately. 
  
Finally, we give improved $2$-approximations for the empty tours variants (\Cref{alg:mamvtsp0_2apx,alg:mdmvtsp0_2apx}).
These approaches build on the techniques from B{\'e}rczi et al.~\cite{BercziMV2020}, that were used to obtain a polynomial-time $\nicefrac32$\hyp{}approximation for the metric Many-Visits Path TSP.
The algorithms calculate degree-bounded multigraphs, then apply an edge-doubling and an efficient shortcutting procedure.
We provide an overview of these results in \Cref{tab:overview}. 
  
\begin{table}[!ht] \centering
\begin{tabular}{SrScSc}
 \toprule
 & \textbf{arbitrary tours} & \textbf{disjoint tours} \\ \midrule
 unrestricted MV-mTSP\+ & $4$-approx. (\Cref{alg:mamvtsp_3apx}) & $4$-approx. (\Cref{alg:mamvtsp_4apx}) \\ \cline{1-3} 
 unrestricted MV-mTSP\1 & \multicolumn{2}{Sc}{ 
 $2$-approximation (\Cref{alg:mamvtsp0_2apx})} \\ \cline{1-3}
 MV-mTSP\+ & $3$-approx. (\Cref{alg:mdmvtsp_3apx}) & $4$-approx. (\Cref{alg:mdmvtsp_4apx}) \\ \cline{1-3} 
 MV-mTSP\1 & \multicolumn{2}{Sc}{
 $2$-approximation (\Cref{alg:mdmvtsp0_2apx})} \\ 
 \bottomrule
\end{tabular} 
\caption{Overview of approximation algorithms for the different variants of many-visits mTSP.}
\label{tab:overview}
\end{table}    

  The unrestricted MV-mTSP with empty tours generalizes the $m$-cycle cover problem for undirected graphs~\cite{KhachayNeznakhina2016}, hence improving the approximation ratio of \Cref{alg:mamvtsp0_2apx} for symmetric edge costs would imply an improved algorithm for the $m$-cycle cover problem.
  The (single-visit) multidepot mTSP defined in Rathinam et al.~\cite{RathinamSD2007} allows some agents to not visit any cities, i.e.\ having a component in the solution graph consisting of a depot and no other vertices or edges.  
  Thus, the MV-mTSP with empty tours generalizes the multidepot mTSP, and the approximation ratio of \Cref{alg:mdmvtsp0_2apx} matches that of Rathinam et al.~\cite{RathinamSD2007} and is only slightly worse than that of Xu et al.~\cite{XuXR2011}.

Our results do not compare directly to the results by Jansen et al.~\cite{JansenKMR2019} and Deppert and Jansen~\cite{DeppertJansen2019}.
In one aspect, the problem variants considered by the authors can be regarded as easier than the problems considered in this paper, as they assume sequence-independent setup times that only occur before processing a batch of jobs from the same class.
On the other hand, the jobs from the same class considered in those works~\cite{JansenKMR2019, DeppertJansen2019} can have \textit{different} processing times, which is a much more general assumption than the cost structure of jobs considered in this work (the edge costs translate to zero processing times and the sum of processing and setup times being metric).

\section{Preliminaries}

\subsection{Notions and Notations}
A \emph{multigraph} $X$ is a graph on the vertex set $V$ with a multiset $E(X)$ as edge set, that is, $E(X)$ might contain several copies of the same edge.
For a subset $F\subseteq E$ of edges, the \emph{set of vertices covered by $F$} is denoted by $V(F)$.
Given a multigraph $X$, the number of edges leaving the vertex set $C \subseteq V(X)$ is denoted by $\delta_X(C)$.
Similarly, the number of regular edges (i.e. excluding self-loops) in~$X$ incident to a vertex $v\in V$ is denoted by~$\delta_X(v)$.
The number of all edges (i.e. including self-loops) in $X$ incident to a vertex $v \in V$ is denoted by~$\ddelta_X(v)$, where the multiplicity of the self-loop on $v$ is counted twice. 
Given edge costs $c: V \times V \rightarrow \mathbbm{R}_{\geq 0} \cup \{+\infty\}$, the \emph{cost} of a multigraph $X$ is simply the sum of its edge costs, i.e.\ $\cost(X) = \sum_{u,v \in V} x(uv) \cdot c(uv)$, where $x(uv)$ is the number of copies of an edge $uv$ in $X$.

In the many-visits mTSP variants, the input graph $G(V,E)$ contains a set of depots $D \subseteq V$.
In the problem instances throughout this paper, we denote by $\V := V \rep D$ the set of non-depot vertices (or cities).
The request function $r(\cdot)$ is defined only on the set $\V$.

A feasible many-visits mTSP solution is a collection of connected multigraphs.
Observe that each such multigraph $X$ can be regarded as an MVTSP tour, as $X$ is connected and every vertex has an even degree in $X$.
Denote the total number of visit requirements by $r(V) = \sum_{v \in \V} r(v)$.
In order to avoid a linear dependence on $r(V)$ in the time and space complexities, we use a \emph{compact representation} of the multigraph $X$.
We store $X$ using $O(n^2 \log r(V))$ space, by storing each edge $uv$ of $X$ along with its multiplicity $x(uv)$.
Given such a representation of~$X$, an MVSTP tour $T$ with edge set $X$ can easily be recovered, such that $\cost(T) = \cost(X)$, using Hierholzer's algorithm~\cite{Hierholzer1873,Fleischner1991} in $O(r(V)^3)$ time.
Note however that $r(V)$ might be exponential in the input size, hence we will work with the compact representation throughout the paper.

\subsection{Reductions among problem variants}

When it comes to the problem variants with depots, we allow any solution to problems (P5)--(P8) to visit a depot an arbitrary number of times (but requiring a positive number of visits in (P5) and (P6)), although it is not difficult to see that one visit is always sufficient due to the cost function being metric.

\begin{lemma}
\label{lem:depotonce}
  Let a multigraph $X$ be an optimal solution to any of the problems (P5)--(P8).
  Then there exists an optimal solution $X'$ with $\cost(X') \leq \cost(X)$, such that $X'$ visits each depot $d \in D$ at most once.
\end{lemma}
\begin{proof}
  If $X$ visits each depot at most once, we are done.
  Clearly, we may assume that $X$ contains no self-loops on the depots.
  Let us assume that $X$ visits a depot $d \in D$ multiple times, i.e. $\delta_X(d) > 2$, and let $i$ be the agent corresponding to depot $d$. 
  If another agent $j$ visits $d$ as well, then the tour of $j$ can be shortcut at $d$, and such a step does not increase the total cost of the solution because of the cost function being metric.
  Analogously, if the tour of agent $i$ contains more than $2$ edges incident to $d$, then all of these edges except the first and the last ones can be eliminated by applying shortcuts.
  This concludes the proof of the lemma.
\end{proof}

One might wonder whether it is indeed needed to address both the unrestricted variants and the variants with depots, or maybe there exists a reduction between the two in either direction.
Given a MV-mTSP instance on $n$ vertices and $m$ depots, a straightforward way to formulate an unrestricted counterpart is to define a graph on $n+m$ vertices, where the original vertices inherit the visit requests, and the vertices corresponding to depots are given a request of $1$ (according to \Cref{lem:depotonce}, this is sufficient).
However, the depots are needed to end up in different tours, and there is no tool in the Multiple-Agent setting that would enforce such a constraint.
In the other direction, given an unrestricted MV-mTSP instance on $n$ nodes, one could construct a MV-mTSP counterpart by simply keeping the $n$ vertices and adding $m$ depots.
Now the difficulty comes from the fact that the edge costs between the original vertices and the depots should be defined in such a way that an optimal solution to the new MV-mTSP instance corresponds to an optimal solution to the original unrestricted MV-mTSP instance, and the resulting cost function satisfy the triangle inequality.
Such a cost function can be found by guessing one edge in each of the~$m$ tours, but that leads to an algorithm with a running time exponential in $m$. 
As we assume that~$m$ is part of the input, we skip the details here and give direct algorithms instead. 

First, let us show a reduction between the problem variants we consider.
\begin{lemma}
\label{lem:sumc_reduction}
  There are polynomial-time reduction to problems P1, P2, P5 and P6 from problems P3, P4, P7 and P8, respectively.
\end{lemma}
\begin{proof}
  From the problem definitions, an optimal solution $X^\star$ to P3, P4, P7 or P8 consists of~$m$, possibly empty MVTSP tours, denoted by $X^*_1, \dots, X^*_m$.
  Let us reorder the agents such that $X^*_1, \dots, X^*_\ell$ denote the empty tours and $X^*_{\ell+1}, \dots, X^*_m$ the non-empty tours.
  It is easy to see that $X^\star$ is not only an optimal solution to P3, P4, P7 or P8, but their counterparts P1, P2, P5 or P6, with $\ell$ agents.
  This means that solving problem instances of P1, P2, P5 or P6 with $\ell$ agents for $\ell := 1, \dots, m$, the cheapest of the $m$ solutions will be an optimal solution to P3, P4, P7 or~P8.
\end{proof}

The next claim is an easy observation that follows from the problem definitions.
\begin{claim}
\label{clm:sumc_feasible}
  A feasible solution for P1, P2, P5 and P6 is a feasible solution for P3, P4, P7 and~P8, respectively.
  A feasible solution for P2, P4, P6 and P8 is a feasible solution for P1, P3, P5 and~P7, respectively.
\end{claim}  
\begin{proof}
  Problems P3, P4, P7 and P8 only differ from their respective counterparts of P1, P2, P5 and P6 in one detail, namely the former problems allow for empty tours in the solution.
  Moreover, problems P1, P3, P5 and P7 are relaxations of problems P2, P4, P6 and P8, respectively, with the disjointness constraint lifted.
\end{proof}

In what follows, we discuss the similarities and the differences between the problems.
The cost of an optimal solution to problem P$i$ is denoted by $\OPT_{\text{P}i}$ for $i=1,\dots,8$. 
Below, we use the symbol \enquote*{$\equiv$} to indicate that given the same parameters $G$, $c$, $r$ and $m$, the optimal solution values of the two problems in question coincide.
Intuitively, this means that restricting the tours to be disjoint or forbidding empty tours in the solution might have no effect on the optimal solution value.
Similarly, we use the symbol \enquote*{$\not\equiv$} if there are instances where the costs of optimal solutions of the two problems are different.

In the next two claims we prove that, in case we allow the agents to have empty tours, it does not matter whether we require vertex-disjoint tours or allow the tours to overlap.
\begin{claim}
\label{clm:three_four}
  P3 $\equiv$ P4.
\end{claim}
\begin{proof}
  We will prove that the optimal solutions to P3 and P4 have the same optimum cost for a given instance $(G, c, r, k)$.
  
	Let $\X$ be an optimal solution to P3 with cost $\OPT_{\text{P3}}$, and suppose that the tours of agents~$i$ and~$j$ overlap, that is, $V(X^\star_i) \cap V(X^\star_j) \neq \emptyset$.
  As $X^\star_i$ and $X^\star_j$ are MVTSP tours, every vertex $v \in V(X^\star_i)$ has an even degree in $X^\star_i$, and every vertex $v \in V(X^\star_j)$ has an even degree in $X^\star_j$.
  This implies that every $v \in V(X^\star_i) \cup V(X^\star_j)$ has an even degree in $X^\star_i \cup X^\star_j$, and moreover, $X^\star_i \cup X^\star_j$ is connected.
  One can now redefine the tour of agent~$i$ as $X_i := X^\star_i \cup X^\star_j$ and the tour of agent $j$ as the empty graph ($X_j := \emptyset$), making the tours of agents $i$ and $j$ disjoint.
  The resulting multigraph $X$ is a solution to P4, since it has at most $m$ components and visit each vertex $v \in V$ exactly $r(v)$ times. 
  Therefore $\cost(X) \geq \OPT_{\text{P4}}$.
  However, the cost of $X$ is $\OPT_{\text{P3}}$, since they have the same edge multiset. 
  This shows that $\OPT_{\text{P3}} \geq \OPT_{\text{P4}}$ holds.
  
  Now assume that $Y^\star$ is an optimal solution to P4, with cost $\OPT_{\text{P4}}$.
  According to \Cref{clm:sumc_feasible}, the multigraph $Y^\star$ is a feasible solution to P3, thus $\OPT_{\text{P3}} \leq \OPT_{\text{P4}}$ holds, concluding the proof.
\end{proof}

A similar reasoning works for the analogous problems with depots:
\begin{claim}
  \label{clm:seven_eight}
  P7 $\equiv$ P8.
\end{claim}
\begin{proof}
  Let $\X$ be an optimal solution to P7 with cost $\OPT_{\text{P7}}$, and suppose that the tours of agents~$i$ and~$j$ overlap in $\X$; i.e.\ there are $m' < m$ components in $\X$ with at least one component having more than one depots.
  Denote the component with depots $d_i$ and $d_j$ by $X$.
 
  The main idea is to disconnect the depot $d_j$ from $X$ by applying shortcuts.
  We do this by decomposing $X$ into cycles, and considering each cycle $C$ containing $d_j$.
  In each such cycle $C$, we replace the edges $d_j u$ and $d_j v$ by $uv$.
  Due to the possibly exponential number of occurrences of $d_j$ in cycles one wants to refrain performing these edge replacements one-by-one.
  It is possible to decompose a multigraph into a polynomial number of cycles with multiplicities; this is equivalent of finding a cycle-decomposition of a circulation, this is described in detail in the proof of \Cref{lem:shortcuts}.
  Thus we can perform these replacements by changing the multiplicity of a polynomial number of edges, yielding to an efficient reduction.
  
  As a result of the operations above, $d_j$ becomes an isolated vertex and vertices $u$ and $v$ have the same degree as before.
  Moreover, because of the triangle inequality, the transformation do not increase the cost of~$\X$.
  If we repeat this operation until there is exactly one depot in each component of~$\X$, then the resulting multigraph $\X$ is a feasible solution to P8 with cost at most $\OPT_{\text{P7}}$.
  Therefore $\OPT_{\text{P7}} \geq \cost(\X) \geq \OPT_{\text{P8}}$ holds.

  Now assume that $Y^\star$ is an optimal solution to P8, with cost $\OPT_{\text{P8}}$.
  According to \Cref{clm:sumc_feasible}, the multigraph $Y^\star$ is a feasible solution to P7, therefore $\OPT_{\text{P7}} \leq \OPT_{\text{P8}}$ holds, concluding the proof.
\end{proof}

\begin{figure}[!ht]
  \centering
  \includegraphics[scale=0.65]{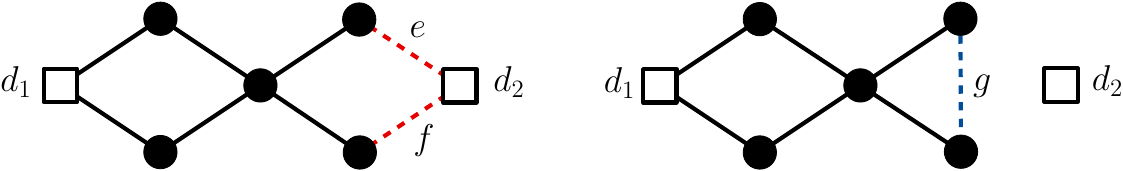}
  \caption{
  Illustration of \Cref{clm:seven_eight}, where the visit requirement of the middle vertex is $2$, while for the rest of the vertices it is $1$.
  Because of the triangle inequality, $c(g)\leq c(e)+c(f)$.
  The two tours overlap in the middle vertex.}
\end{figure}

Based on \Cref{clm:three_four,clm:seven_eight}, we do not distinguish between problems P3 and P4, as well as between problems P7 and P8.
For this reason we will just refer to these problems as unrestricted MV-mTSP with empty tours (unrestricted MV-mTSP\0) and MV-mTSP with empty tours (MV-mTSP\0), respectively.

On the other hand, if we do not allow empty tours, then the optimal solutions for the arbitrary and the disjoint tours variants can be different.

\begin{claim}
\label{clm:onetwo_fivesix}
  P1 $\not\equiv$ P2, P5 $\not\equiv$ P6.
\end{claim}
\begin{proof}
  According to \Cref{clm:sumc_feasible}, solutions to problems P2 and P6 are feasible solutions to P1 and P5, respectively.
  Moreover, the only constraint that problems P2 and P6 impose on the feasible solutions in addition to the constraints of P1 and P5 is that the tours of the agents have to be disjoint.
  We show that the disjointness requirement is meaningful, i.e.\ an optimal solution to an instance of problem P2 or P6 can have strictly higher cost than an optimal solution to the corresponding instance of P1 or P5, respectively.

  In case of problems P5 and P6, the example in \Cref{fig:onetwo_fivesix} shows that, by appropriately setting the edge costs, the solution on the right has significantly higher cost than the solution on the left.
  For problems P1 and P2 the same example demonstrates the claim, by assuming that $d_1$ and $d_2$ are regular (non-depot) vertices with a request of $1$.  
\end{proof}

\begin{figure}[h]
  \centering
  \includegraphics[scale=0.65]{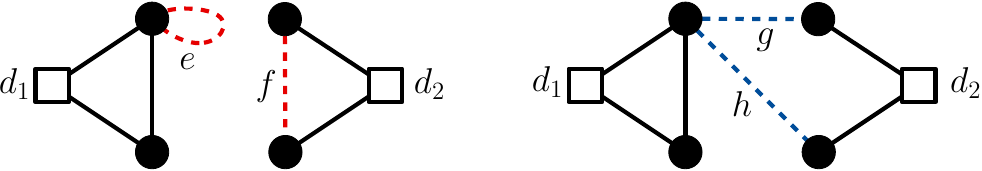}
  \caption{
  Illustration of \Cref{clm:onetwo_fivesix}, where the visit requirement of the top left vertex is $2$, while for the rest of the vertices it is $1$.
  Because of the triangle inequality, $c(e)\leq c(g)+c(h)$ and $c(f) \leq c(g)+c(h)$.
  If $c(e)+c(f)> c(g)+c(h)$ holds, then $g$ and $h$ together costs less than edges $e$ and $f$.}
  \label{fig:onetwo_fivesix}
\end{figure}

Finally, we show that allowing empty tours in the solutions makes a significant difference.
\begin{claim}
\label{clm:two_four}
   P1 $\not\equiv$ P3, P2 $\not\equiv$ P4.
\end{claim}
\begin{proof}
  Our goal is to show that allowing less than $m$ components might result in an optimal solution having cost strictly less than that of one with exactly $m$ components.
  This can happen in both the overlapping tours and the disjoint tours variants.  
  
  Let $\hat{G}=(V,\hat{E})$ be a cycle of length $n:=\abs{V}$. 
  Now define a complete graph $G$ on the same vertex set $V$, where the cost of edge $v_iv_j$ is equal to the shortest path distance between~$v_i$ and~$v_j$ in $\hat{G}$.
  Moreover, let $r(v)=1$ for all $v \in V$.
    For any $m \in \mathbbm{Z}_{\geq 1}$, the optimal solution for the unrestricted MV-mTSP\0 is the cycle $v_1v_2,v_2v_3,\dots,v_{n-1}v_n,v_nv_1$ with cost $n$.
  However, if we require exactly $m$ components in our solution, the optimal solution has to use edges with costs larger than~$1$, resulting in a solution with a cost larger than~$n$.
\end{proof}

An analogous statement is true when depots are present:
\begin{claim}
\label{clm:six_eight}
  P5 $\not\equiv$ P7, P6 $\not\equiv$ P8.
\end{claim}
\begin{proof}
  Consider an instance where the cost of every edge is $1$, independently from whether it goes between two cities or between a depot and a city.
  It is not difficult to verify that an optimal solution to MV-mTSP\0 with disjoint tours consists of one large connected component and $m-1$ isolated depots.
  Therefore the optimal cost (that is, the number of edges in the solution) is $\sum_{v \in V} r(v) + 1$.
  On the other hand, according to \Cref{lem:depotonce}, in an optimal solution to MV-mTSP\+ with disjoint tours the degree of every depot $d \in D$ is $2$, therefore the optimal cost is $\sum_{v \in V} r(v) + m$, which is larger than the optimal cost of MV-mTSP\0 with disjoint tours whenever $m>1$. 
  The same reasoning holds if we allow the tours to overlap.
\end{proof}

\begin{remark}  
  It is worth mentioning that the problems do not always admit a feasible solution. This happens when empty tours are not allowed, and either the total visit requirements $\sum_{v\in V} r(v)$ is strictly less than the number of agents/depots in the arbitrary setting, or the number of cities is strictly less than the number of agents/depots in the disjoint setting. 
\end{remark}

\section{Approximation algorithms for many-visits mTSP with arbitrary tours}
In this section, we provide a $3$- and a $4$-approximation to MV-mTSP\+ and unrestricted MV-mTSP\+ with arbitrary tours, respectively, and due to \Cref{lem:sumc_reduction}, these results carry over to MV-mTSP\0 and unrestricted MV-mTSP\0 as well.
The algorithms share the same basic ideas, hence we first present the notions and techniques that are common to all algorithms.
Throughout \Crefrange{lem:tp_relax_ma}{lem:tp_relax_md}, we assume that the cost function $c$ satisfies the triangle inequality.

\paragraph{The transportation problem.}
Observe that if we relax the connectivity and decomposability constraints from the unrestricted MV-mTSP variants (similarly to the single-salesman case~\cite{CosmadakisPapadimitriou1984,BergerKMV2020,BercziMV2020}), the problem can be modelled as the Hitchcock transportation problem~\cite{Hitchcock1941}.
The transportation problem can be solved efficiently using a min-cost max-flow algorithm~\cite{EdmondsKarp1970}, even in strongly polynomial time~\cite{Orlin1993, KleinschmidtSchannath1995}.
Due to the problem definition, the cost of the optimal transportation problem solution provides a lower bound on the corresponding unrestricted MV-mTSP optimal solutions:
  
\begin{lemma}
\label{lem:tp_relax_ma}
  Let $\X$ be an optimal solution to any of the unrestricted MV-mTSP variants (arbitrary, disjoint or empty tours) on an instance $(G, c, r)$.
  Let $\TP$ be an optimal transportation problem solution on $(G, c, r)$.
  Then, $\cost(\TP) \leq \cost(\X)$. \hfill \qed
\end{lemma}

For request vectors $r, r': V \rightarrow \mathbbm{Z}_{\geq 1}$, let us denote by $r \preccurlyeq r'$ if $r(v) \leq r'(v)$ holds for all~$v \in V$.
\begin{lemma}
\label{lem:tp_compare}
  Let $\TP$ and $\TP'$ be optimal transportation problem solutions for the instances $(G, c, r)$ and $(G, c, r')$, respectively.
  If $r \preccurlyeq r'$ holds, then $\cost(\TP) \leq \cost(\TP')$.
\end{lemma}
\begin{proof}
  Observe that $\ddelta_\TP(v) = 2 \cdot r(v)$ and $\ddelta_{\TP'}(v) = 2 \cdot r'(v)$ hold by definition.
  If $\ddelta_\TP(\cdot) \equiv \ddelta_{\TP'}(\cdot)$, we are done.
  Otherwise, there is at least one vertex $v \in V(G)$, such that $\ddelta_\TP(v) < \ddelta_{\TP'}(v)$.
  Apply shortcuts in~$\TP$, until $\ddelta_\TP(v) = \ddelta_{\TP'}(v)$ holds for every $v \in V(G)$ by performing the following two operations in total of $\nicefrac{(\ddelta_{\TP'}(v)-\ddelta_\TP(v))}{2}$ times:
  \begin{enumerate}[label=(\Alph*)]
     \item If $v$ is incident to copies of self-loops $vv$, remove a copy of $vv$.
     \item If there exist not necessarily distinct vertices $u, w$ adjacent to $v$, replace a copy of $uv$ and $vw$ by a copy of $uw$.
  \end{enumerate}
  Operation (A) does not increase the cost of $\TP'$, as the edge costs are non-negative.
  Due to the triangle inequality, operation (B) does not increase the cost of $\TP$ either.
  Therefore, the cost of the resulting multigraph $\TP''$ will be $\cost(\TP'') \leq \cost(\TP')$.   
  Moreover, since both operations (A) and (B) decrease the degree of a vertex by $2$, applying the operations $\nicefrac{(\ddelta_{\TP'}(v)-\ddelta_\TP(v))}{2}$ times will result in $v$ having a degree of $\ddelta_\TP(v)$ in $\TP''$.
 
  Perform the operations above for every vertex $v \in V$.
  Now the multigraph $\TP$ has the same degree sequence as $\TP''$. Moreover, $\TP$ is a minimum cost such multigraph.
  Hence $\cost(\TP) \leq \cost(\TP'')$ holds, and the lemma follows.
\end{proof}

Now we show that the claim of \Cref{lem:tp_relax_ma} also holds for the problem variants with depots.

\begin{lemma}
\label{lem:tp_relax_md}
  Let $\X$ be an optimal solution to any of the MV-mTSP variants (arbitrary, disjoint or empty tours) on an instance $(G, c, r)$.
  Let $\TP$ be an optimal transportation problem solution on $(G, c, r)$.
  Then, $\cost(\TP) \leq \cost(\X)$.
\end{lemma}
\begin{proof}
  Let us define an auxiliary graph $G'$ on the vertex set $V$ with edge costs $c'(uv) := c(uv)$ for every vertex $u, v \in V$.
  Define the requests of vertices $v \in \V$ by $r'(v) := r(v)$, and the requests of the depots $r'(d) := \ddelta_{X^\star}(d)/2$ for $d \in D$.
  Calculate an optimal transportation problem solution $\TP'$ on $(G', c', r')$.
  
  $\TP'$ is a minimum cost multigraph with the same degree sequence as $\X$, while the connectivity and decomposability constraints are relaxed.
  Hence $\cost(\TP') \leq \cost(\X)$.
  Moreover, as the depot set $D$ is not covered by $\TP$, one can think of the requests of $d \in D$ in the instance $(G,c,r)$ as $r(d)=0$.
  Since $r(v)=r'(v)$ for vertices $v \in \V$ and $r(d) \leq r'(d) = \ddelta_\X(d)$ for $d \in D$, $r \preccurlyeq r'$ holds.
  Due to \Cref{lem:tp_compare}, $\cost(\TP) \leq \cost(\TP')$ follows, which, together with $\cost(\TP') \leq \cost(\X)$, proves the claim.
\end{proof}

As a direct consequence of \Crefrange{lem:tp_relax_ma}{lem:tp_relax_md}, we have the following corollary.

\begin{corollary}
\label{lem:tp_relax}
  Let $(G,c,r)$ be an instance of any unrestricted MV-mTSP variants \textbf{or} any MV-mTSP variants, and let $\X$ be an optimal solution to this instance.
  Let $\TP'$ be an optimal solution to the transportation problem on $(G, c, r')$ with $r' \preccurlyeq r$.
  Then, $\cost(\TP') \leq \cost(\X)$.
\end{corollary}  

\subsection{Unrestricted many-visits mTSP with arbitrary tours}

  We first consider the problem variants without depots.
 Let $G(V,E)$ be a complete graph, $c(uv)$ be metric and symmetric edge costs, and $r(v)$ be requirements for every vertex $v \in V$.
  In all variants of unrestricted MV-mTSP, we seek a minimum cost multigraph $\X$ that visits each vertex $v \in V$ a total of $r(v)$ times.
We present a simple $4$-approximation for the unrestricted MV-mTSP\+ with arbitrary tours.

\begin{algorithm}[!ht]
  \caption{A $4$-approximation for the unrestricted MV-mTSP\+ with arbitrary tours  \label{alg:mamvtsp_3apx}}
  \begin{algorithmic}[1]
    \Statex \textbf{Input:} A complete undirected graph $G(V,E)$, costs $c:E\rightarrow\mathbbm{R}_{\geq 0}$ satisfying the triangle inequality, requests $r:V\rightarrow\mathbbm{Z}_{\geq 1}$, number of agents $m$.
    \Statex \textbf{Output:} $m$ non-empty MVTSP tours that visit each $v \in V$ a total of $r(v)$ times, or NO if no solution exists. 
    \State If $m>r(V)$, then \textbf{return} NO.
    \State If $m>n$, then for each $v\in V$ add $r(v)$ copies of the self-loop $vv$ to $X$. Take a partition $X:=X_1\cup\dots\cup X_m$ of $X$ into $m$ non-empty sets such that $X_i$ contains copies of the same self-loop for each $i=1,\dots,m$, and \textbf{return }$X := X_1 \cup \dots \cup X_m$.
    \State  Determine a minimum cost spanning forest $F$ of $G$ consisting of $m$ components, and let $F_1, \dots, F_m$ denote its components.\label{st:simple_ma_forest}
    \State For $i=1,\dots,m$, if $\abs{V(F_i)}\geq 2$ then let $H_i$ denote the cycle on $V(F_i)$ obtained by duplicating the edges of $F_i$ and applying shortcuts, otherwise let $H_i$ consist of a single copy of the loop on the vertex in $V(F_i)$. \label{st:simple_ma_cycles}
    \State For each $v\in V$, pick an index $i$ with $v \in V(H_i)$, and add $r(v)-1$ copies of the self-loop $vv$ to $H_i$. Denote the resulting multigraphs by $X_i$ for $i=1,\dots,m$.\label{st:simple_ma_loops}
   \Statex  \textbf{return} $X := X_1 \cup \dots \cup X_m$
  \end{algorithmic}
\end{algorithm}

\begin{theorem}
\label{thm:mamvtsp_3apx}
\Cref{alg:mamvtsp_3apx} provides a $4$-approximation for unrestricted MV-mTSP\+ with arbitrary tours in time polynomial in~$n$,~$m$ and~$\log r(V)$.
\end{theorem}
\begin{proof}
  We first prove that \Cref{alg:mamvtsp_3apx} constructs a feasible solution to the unrestricted MV-mTSP\+ with arbitrary tours, and provide an upper bound to its cost.
  Finally, we analyze the time complexity.
  
\paragraph{Feasibility.}
  The cycles $H_i$ contribute a degree of $2$ while the loops contribute a degree of $2 \cdot (r(v)-1)$ to the degree of each $v \in V$. Therefore the total degree of each vertex $v$ in $X_1 \cup \dots \cup X_m$ is $2 \cdot r(v)$.
    
\paragraph{Cost of solution.}
  Let us denote by $m^\star$ the number of \textit{components} (which does not necessarily equal the number of MVTSP tours) in the optimal solution $\X$, where $1 \leq m^\star \leq m$.
  An optimal solution $\X$ contains a spanning forest~$F^\star$ having $m^\star$ components with $\cost(F^\star)$ being at least $\cost(\MSF^{m^\star})$, where $\MSF^{m^\star}$ denotes a minimum cost spanning forest having $m^\star$ components.
  It is easy to see that $\cost(F) \leq \cost(\MSF^{m^\star})$, because~$F$ is a minimum cost spanning forest with $m \geq m^\star$ components.  
  Hence $\cost(F) \leq \cost(F^\star) \leq \cost(X^\star)$ follows.
  
  Since the edge costs are metric, the shortcutting operations in \Cref{st:simple_ma_cycles} do not increase the cost of the multigraph, therefore the total cost of the \enquote*{non-loop $H_i$}s is at most $2 \cdot \cost(\X)$.
  We claim that the total cost of the self-loops added in \Cref{st:simple_ma_cycles,st:simple_ma_loops} is also at most $2 \cdot \cost(\X)$.
  Indeed, we have
\begin{equation}
\label{eq:loops_cost}
\sum_{v \in V} r(v)\cdot c(vv)\leq\sum_{v \in V} r(v)\cdot 2 \cdot \cmin(v) \leq 2 \cdot \cost(\TP^\star_{c,r}) \leq 2 \cdot \cost(\X),
\end{equation}
where the second inequality follows from the fact that ensuring a visit to a vertex $v$ costs at least $\cmin(v)$, and one needs $r(v)$ many of such visits.
  This verifies the approximation ratio.

\paragraph{Complexity analysis.}
Calculating a minimum cost spanning forest consisting of $m$ components can be done in time polynomial in~$n$ and $\log r(V)$.
  Throughout the algorithm, we use a compact representation of all multigraphs, this way the time and space complexity of graph operations can be bounded by $O(n^2 \log r(V))$, hence the remaining graph operations can also be done in polynomial time.
\end{proof}

By \Cref{lem:sumc_reduction}, an analogous result holds for the unrestricted MV-mTSP\0.

\begin{corollary}
\label{cor:mamvtsp_3apx}
There exists a $4$-approximation algorithm for unrestricted MV-mTSP\0 with running time polynomial in~$n$,~$m$ and~$\log r(V)$.
\end{corollary}

\begin{remark}
 A natural idea would be to -- instead of doubling the edges of the forest -- calculate a minimum cost matching on the odd degree vertices of $F$ in \Cref{alg:mamvtsp_3apx}.
  However, such an approach might lead to a multigraph that cannot be decomposed into $m$ tours.
  Consider the case where all requirements are~$1$, the number of agents is~$m=2$, and~$F$ is the union of two disjoint paths.
  In this case the odd-degree vertices in $F$ are the four endpoints of the paths, and a minimum cost matching might connect them into a single cycle.
  Such a cycle cannot be decomposed into more tours than one, therefore it is not a feasible solution to the unrestricted MV-mTSP\+ with arbitrary tours.
  \end{remark}

\subsection{Many-visits mTSP with arbitrary tours}

  In the MV-mTSP variants, a set $D \subseteq V$ of depots given; let us define $m:=\abs{D}$.
  Recall that we denote by $\V := V \rep D$ the set of non-depot vertices (or cities), and that $r(\cdot)$ is defined only on~$\V$.
Our goal is to give a simple $3$-approximation for the MV-mTSP\+ with arbitrary tours.

  Due to the restrictions involving depots, we start with building a special spanning multigraph instead of a forest. 
 Given a set~$D$ of special vertices, Cerdeira~\cite{Cerdeira1994} provides a matroid-based algorithm of calculating a minimum cost forest such that every component contains exactly one vertex from $D$ and none of the components is trivial, i.e.\ they contain at least one vertex from $\V$.
Unfortunately, we cannot use Cerdeira's algorithm~\cite{Cerdeira1994} directly on the vertex set $V$.
  Note that the MV-mTSP\+ with arbitrary tours allows the MVTSP tours to overlap in one or more vertices.
  This means that an optimal solution may not contain a spanning forest with the properties of Cerdeira's solution~\cite{Cerdeira1994}, but only a spanning forest with at least one (but possibly more) depot in each component.
  Therefore, we cannot compare the cost of the spanning forest calculated by the algorithm of Cerderia~\cite{Cerdeira1994} to the cost of an optimal solution.\footnote{
  Moreover, the MV-mTSP\+ with arbitrary tours always have a feasible solution if $\abs{D} \leq \sum_{v \in V} r(v)$, even if the depots outnumber the vertices, i.e.\ $\abs{D} > \abs{V}$.
  In these cases, a spanning forest with the aforementioned properties does not even exist.}
  
  For the reasons above, we first build an auxiliary graph $G'$ in a way that a constrained spanning forest of $G'$ obtained by the algorithm of Cerdeira~\cite{Cerdeira1994} yields an appropriate spanning multigraph in $G$ -- one that allows more than one vertices from $D$ in a component but has a cost at most the optimal solution to the MV-mTSP\+ with arbitrary tours.

\paragraph{Spanning multigraph of $G$.}  
  Let $\hat{G}$ be an auxiliary graph consisting of the depot set~$D$ and $m(v) := \min\{\abs{D}, r(v)\}$ copies of every vertex $v \in \V$, denoted by $v_1, \dots, v_{m(v)}$.
  We denote the set of these vertices by $\hat{V}$.
  The edges between the copies of the same vertex $v$ incur no cost, i.e.\ $\hat{c}(v_i,v_j) = 0$ for $i, j \in 1, \dots, m(v)$, while the other edge costs are inherited from $G$, that is, $\hat{c}(d,v_i) := c(d,v)$ and $\hat{c}(v_i,w_j) := c(v,w)$ for all copies $i=1,\dots,m(v)$ and $j=1,\dots,m(w)$ of $v$ and $w$, respectively.
  We calculate a minimum cost spanning forest $\hat{F}$, such that each component of $\hat{F}$ contains exactly one depot $d \in D$ and at least one vertex $v \in \hat{V}$, using the algorithm by Cerdeira~\cite{Cerdeira1994}.
  Then the forest~$\hat{F}$ is transformed into a spanning multigraph $F$ on the vertex set~$V$ by identifying all copies $v_1, \dots, v_{m(v)}$ of~$v$ with the single vertex~$v$.
  
\begin{lemma}
  \label{lem:aux_forest}
  The graph $F$ calculated above is a minimum cost spanning multigraph on~$G(V,E)$ such that
  \begin{enumerate}
  \itemsep0em    
  \item[1)]  every component of $F$ contains at least one depot $d \in D$,
  \item[2)]  every component of $F$ contains at least one vertex $v \in \V$,
  \item[3)]  no vertex $v \in \V$ has more than $r(v)$ depots among its neighbours.
  \end{enumerate}
\end{lemma}
\begin{proof}
  Every component of $\hat{F}$ contains exactly one depot. 
  Transforming~$\hat{F}$ to~$F$ might merge components but cannot divide them, therefore the components of~$F$ will contain at least one depot.
  An analogous argument proves that property \textit{2)} holds for a component of~$F$ as well.
  Finally, due to construction, a copy $v_i$ of a vertex $v \in \V$ can be adjacent to at most~1 depot in~$\hat{F}$, and there are at most $r(v)$ copies of a vertex $v$.
  This proves property \textit{3)}.
  
  Now let us turn to the optimality of $F$.
  Note that while transforming $\hat{F}$ into $F$, we delete edges with zero cost, hence $\cost(\hat{F}) = \cost(F)$.
  Suppose indirectly that there exists a graph~$F'$ satisfying properties \textit{1)}--\textit{3)} of lower cost than that of $F$, i.e.\ $\cost(F') < \cost(F)$.
  Then, from $F'$ we can construct a constrained spanning forest on $\hat{G}$ with exactly one depot and at least one regular vertex in each component as follows.
  Make $m(v)-1$ additional copies $v_1, \dots, v_{m(v)-1}$ of each vertex $v \in \V$, and connect them to $v$ in an arbitrary way, then rename the vertex $v$ into $v_{m(v)}$.
  The resulting graph contains a constrained spanning forest $\hat{F}'$ on $\hat{G}$ with the same properties as $\hat{F}$.
  Moreover, because the transformation from $F'$ into $\hat{F'}$ only added edges with zero cost, $\cost(\hat{F}') \leq \cost(F')$ holds.
  From the indirect assumption $\cost(F') < \cost(F)$ it follows that $\cost(\hat{F}') < \cost(\hat{F})$ holds, which contradicts the optimality of~$\hat{F}$.
\end{proof} 

  We will use the procedure above as a subroutine to get the approximation algorithm in \Cref{alg:mdmvtsp_3apx}, see \Cref{tab:3apx} for an illustration.
  
  \begin{algorithm}[!ht]
  \caption{A $3$-approximation for the MV-mTSP\+ with arbitrary tours  \label{alg:mdmvtsp_3apx}}
  \begin{algorithmic}[1]
    \Statex \textbf{Input:} A complete undirected graph $G(V,E)$, costs $c:E\rightarrow\mathbbm{R}_{\geq 0}$ satisfying the triangle inequality, depot set $D \subseteq V$, number of agents $m=\abs{D}$, requests $r:\V\rightarrow\mathbbm{Z}_{\geq 1}$.
    \Statex \textbf{Output:} $m$ non-empty MVTSP tours that visit each $v \in \V$ a total of $r(v)$ times, or NO if no solution exists.
    \State If $m>r(\V)$ then \textbf{return} NO.
    \State  Determine a minimum cost spanning forest $\hat{F}$ of $\hat{G}$, as described before \Cref{lem:aux_forest}, that consists of~$m$ nontrivial components $\hat{F}_1,\dots,\hat{F}_m$ with exactly one depot in each component~\cite{Cerdeira1994}.  \label{st:mdmvtsp_3_forest}
    \State  Identify all copies $v_1, \dots, v_{m(v)}$ of~$v \in \hat{G}$ into a single vertex~$v$, denote the image of $\hat{F}$ by $F$, and let~$F_i$ be the graph arising from $\hat{F}_i$ for $i=1,\dots,m$.\label{st:mdmvtsp_3_identify}
    \State For $i=1,\dots,m$, duplicate the edges of $F_i$ and apply shortcuts to obtain a Hamiltonian cycle $H_i$ on $V(F_i)$.\label{st:mdmvtsp_3_duplicate}
    \State  Let $X_i$ for all $i=1, \dots, m$ be the MVTSP tours of the agents, with the initial values $X_i := H_i$.
    \State  Determine an optimal solution $\TP$ to the transportation problem defined on $G[V]$ with supply and demand being equal to $r'(v):=r(v)-\abs{\{i\mid v\in V(F_i)\}}$ for $v\in \V$. \label{st:mdmvtsp_3_transport}
    \State For each component $\TP_j$ of $\TP$, take an agent $i$ with $V(\TP_j) \cap V(H_i) \neq \emptyset$, and update $X_i:=X_i\cup \TP_j$.\label{st:mdmvtsp_3_assign_tp}
    \Statex  \textbf{return} $X := X_1 \cup \dots \cup X_m$
  \end{algorithmic}
\end{algorithm}

\begin{theorem}
\Cref{alg:mdmvtsp_3apx} provides a $3$-approximation for MV-mTSP\+ with arbitrary tours in time polynomial in~$n$,~$m$ and~$\log r(V)$.
\end{theorem}
\begin{proof}
  Similarly to the unrestricted setting, we first prove that \Cref{alg:mdmvtsp_3apx} constructs a feasible solution to the restricted problem with arbitrary tours, and provide an upper bound to its cost.
  Finally, we analyze the time complexity.

\paragraph{Feasibility.}
  
  The Hamiltonian cycles constructed in \Cref{st:mdmvtsp_3_duplicate} contribute a degree of $2\cdot(r(v)-r'(v))$ to the degree of each vertex $v \in \V$. 
  Furthermore, by construction, each $v$ in $\V$ is contained in at most $r(v)$ components of $\hat{F}$, therefore $r'(v)\geq 0$ holds and the transportation problem has a solution. 
 The solution $\TP$ contributes with a degree of $2 \cdot r'(v)$ to the degree of each $v \in \V$. Therefore the total degree of each vertex $v$ in $X_1 \cup \dots \cup X_m$ is $2 \cdot r(v) $.
      
\paragraph{Cost of solution.}
  Let $\X=\X_1\cup\dots\cup\X_m$ be an optimal solution for the problem.
  By \Cref{lem:tp_relax}, $\cost(\TP) \leq \cost(\X)$ holds.
  We claim that the cost of $F$ is also bounded from above by the cost of $X^\star$.
  Indeed, $\X$ fulfills properties \textit{1)}--\textit{3)} and, by \Cref{lem:aux_forest}, $F$ is a minimum cost spanning multigraph with these properties.
  Due to metric edge costs, the shortcutting operation in \Cref{st:mdmvtsp_3_duplicate} cannot increase the cost of the multigraph, hence the approximation ratio follows.

\paragraph{Complexity analysis.}
  The number of vertices in $\hat{G}$ is $O(n^2)$, therefore \Cref{st:mdmvtsp_3_forest,st:mdmvtsp_3_identify} take polynomial time.
  The operations in \Cref{st:mdmvtsp_3_duplicate,st:mdmvtsp_3_assign_tp} can also be performed efficiently.
  Finally, the transportation problem in \Cref{st:mdmvtsp_3_transport} can be solved in polynomial time as well.
\end{proof}

\begin{figure}[t]
\centering
\begin{subfigure}[b]{0.30\textwidth}
  \centering
  \includegraphics[width=.9\linewidth]{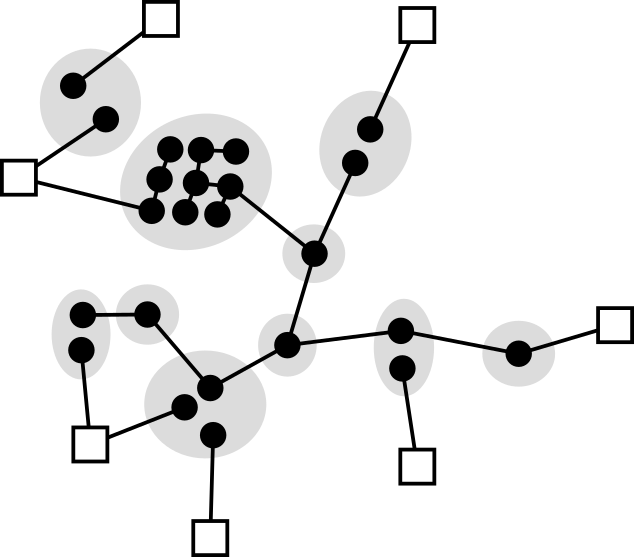}
  \caption{Spanning forest $\hat{F}$ on the auxiliary graph $\hat{G}$ (\Cref{st:mdmvtsp_3_forest}).}
  \label{fig:1}
\end{subfigure}
\hfill
\begin{subfigure}[b]{.30\textwidth}
  \centering
  \includegraphics[width=.9\linewidth]{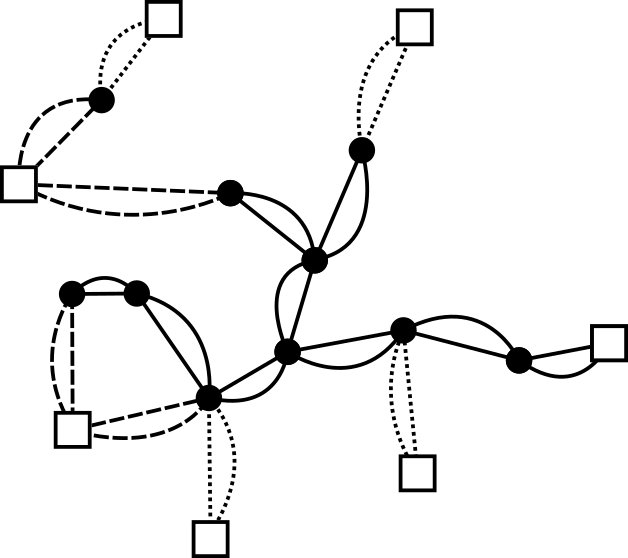}
  \caption{Spanning multigraph $F$, after duplicating the edges (\Crefrange{st:mdmvtsp_3_identify}{st:mdmvtsp_3_duplicate}).} 
  \label{fig:2}
\end{subfigure}
\hfill
\begin{subfigure}[b]{.30\textwidth}
  \centering
  \includegraphics[width=.9\linewidth]{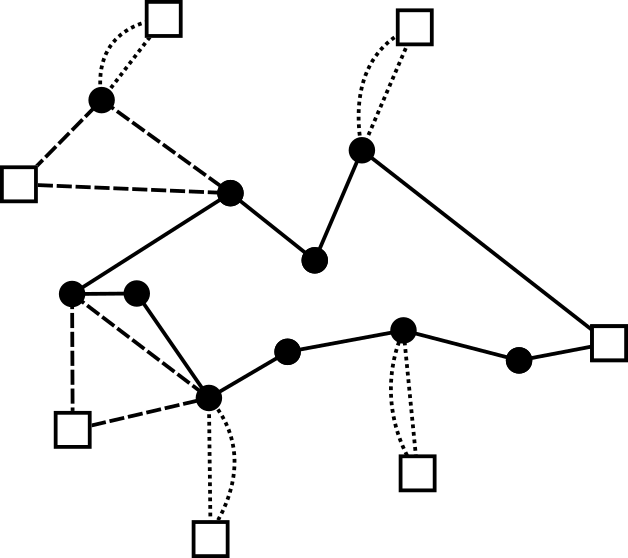}
  \caption{Hamiltonian cycles $H_i$ of the agents (\Cref{st:mdmvtsp_3_duplicate}).}
  \label{fig:3}
\end{subfigure}
\caption{Illustrating \Crefrange{st:mdmvtsp_3_forest}{st:mdmvtsp_3_duplicate} of \Cref{alg:mdmvtsp_3apx}.
Shaded ellipses in \Cref{fig:1} indicate copies of the same vertex $v \in \V$.
Cycles $H_i$ of different agents marked with different edge styles for clarity (\Crefrange{fig:2}{fig:3}).
}
\label{tab:3apx}
\end{figure}

By \Cref{lem:sumc_reduction}, an analogous result holds for the MV-mTSP\0.

\begin{corollary}
\label{cor:mdmvtsp_3apx}
There exists a $3$-approximation algorithm for MV-mTSP\0 with running time polynomial in~$n$,~$m$ and~$\log r(V)$.
\end{corollary}

\begin{remark}
  The difficulties in the MV-mTSP\+ with arbitrary tours is caused by the fact that non-trivial tours are required to contain a depot as well, therefore we cannot compare the cost of an optimal solution with the cost of a simple spanning forest.
  As a workaround, we used an algorithm by Cerdeira \cite{Cerdeira1994} for determining a special spanning multigraph.
  When empty tours are allowed, one can build the spanning forest using the approach of Rathinam et al.~\cite{RathinamSD2007} instead. 
\end{remark}

\section{Approximation algorithms for many-visits mTSP with disjoint tours}

  In this section we show simple $4$-approximations for the unrestricted MV-mTSP\+ with disjoint tours and MV-mTSP\+ with disjoint tours.
  The main difficulty in these variants is that the MVTSP tours of the different agents must be vertex-disjoint.
  This restriction prevents us to use the approach of \Cref{alg:mdmvtsp_3apx}, where we calculated $m$ cycles from the $m$ components of a (constrained) spanning forest, and augmented this graph with a transportation problem solution.
  The transportation problem solution might contain edges with the two endpoints being in different cycles, and thus the resulting multigraph would consist of less than $m$ components.
  In order to prevent such situations, we follow an approach analogous to \Cref{alg:mamvtsp_3apx}, and add copies of self-loops instead of the transportation problem solution.

\begin{algorithm}[!ht]
  \caption{A $4$-approximation for the unrestricted MV-mTSP\+ with disjoint tours  \label{alg:mamvtsp_4apx}}
  \begin{algorithmic}[1]
    \Statex \textbf{Input:} A complete undirected graph $G(V,E)$, costs $c:E\rightarrow\mathbbm{R}_{\geq 0}$ satisfying the triangle inequality, requests $r:V\rightarrow\mathbbm{Z}_{\geq 1}$, number of agents $m$.
    \Statex \textbf{Output:} $m$ vertex-disjoint non-empty MVTSP tours that visit each $v \in V$ a total of $r(v)$ times, or NO if no solution exists.
	\State If $m>\abs{V}$, then \textbf{return} NO.
    \State  Determine a minimum cost spanning forest $F$ of $G$ consisting of $m$ components, and let $F_1,\dots,F_m$ denote its components.\label{st:4apx_ma_forest}
    \State For $i=1,\dots,m$, if $\abs{V(F_i)}\geq 2$ then let $H_i$ denote the cycle on $V(F_i)$ obtained by duplicating the edges of $F_i$ and applying shortcuts, otherwise let $H_i$ consist of a single copy of the loop on the vertex in $V(F_i)$. \label{st:4apx_ma_duplicate}
    \State For each $v\in V$, pick an index $i$ with $v \in V(H_i)$, and add $r(v)-1$ copies of the self-loop $vv$ to $H_i$. Denote the resulting multigraphs by $X_i$ for $i=1,\dots,m$.\label{st:4apx_ma_loops}
    \Statex  \textbf{return} $X := X_1 \cup \dots \cup X_m$
  \end{algorithmic}
\end{algorithm}

\begin{theorem}
\Cref{alg:mamvtsp_4apx} provides a $4$-approximation for unrestricted MV-mTSP\+ with disjoint tours in time polynomial in~$n$,~$m$ and~$\log r(V)$.
\end{theorem}
\begin{proof}
   The feasibility and complexity analysis follows from the same arguments as in the proof of \Cref{thm:mamvtsp_3apx}.
\end{proof}

  By replacing \Cref{st:4apx_ma_forest} of \Cref{alg:mamvtsp_4apx} by the algorithm of Cerdeira~\cite{Cerdeira1994}, we obtain a $4$\hyp{}approximation algorithm for the analogous problem with depots; please refer to \Cref{alg:mdmvtsp_4apx}.

\begin{algorithm}[!ht]
  \caption{A $4$-approximation for the MV-mTSP\+ with disjoint tours  \label{alg:mdmvtsp_4apx}}
  \begin{algorithmic}[1]
    \Statex \textbf{Input:} A complete undirected graph $G(V,E)$, costs $c:E\rightarrow\mathbbm{R}_{\geq 0}$ satisfying the triangle inequality, depot set $D \subseteq V$, number of agents $m=\abs{D}$, requests $r:\V\rightarrow\mathbbm{Z}_{\geq 1}$.
    \Statex \textbf{Output:} $m$ vertex-disjoint non-empty MVTSP tours that visit each $v \in \V$ a total of $r(v)$ times, or NO if no solution exists. 
    \State If $m>\abs{\V}$, then \textbf{return} NO.
    \State  Determine a minimum cost spanning forest $F$ of $G$ consisting of $m$ nontrivial components $F_1,\dots,F_m$ with exactly one depot in each component~\cite{Cerdeira1994}.\label{st:4apx_md_forest}
    \State For $i=1, \dots, m$, duplicate the edges of $F_i$ and apply shortcuts to obtain a Hamiltonian cycles $H_i$ on $V(F_i)$. \label{st:4apx_md_duplicate}
    \State Add $r(v)-1$ copies of the self-loop $vv$ to $H_i$ for each $v \in V(H_i)$, and denote the resulting multigraph by $X_i$.\label{st:4apx_md_loops}
    \Statex  \textbf{return} $X := X_1 \cup \dots \cup X_m$
  \end{algorithmic}
\end{algorithm}

\begin{theorem}
\Cref{alg:mdmvtsp_4apx} provides a $4$-approximation for MV-mTSP\+ with disjoint tours in time polynomial in~$n$,~$m$ and~$\log r(V)$.
\end{theorem}
\begin{proof}
  The feasibility and complexity analysis follows from the same arguments as in the proof of \Cref{alg:mamvtsp_4apx}, by adding that each depot $d \in D$ has an even degree in $H$ and therefore in~$X$ as well.  
  Since an optimal solution $\X$ contains a spanning forest with the properties of those of $F$, and $F$ is a minimum cost such spanning forest, the bound $\cost(F) \leq \cost(\X)$ and thus $\cost(H) \leq 2 \cdot \cost(\X)$ follows.
  \Cref{eq:loops_cost} also holds, proving that $X$ is indeed a $4$-approximation.
\end{proof}

\section{Improved $2$-approximation algorithms for the empty tours variants}

In this section, we provide improved algorithms for unrestricted MV-mTSP\0 and MV-mTSP\0.
For these problems, \Cref{cor:mamvtsp_3apx,cor:mdmvtsp_3apx} provide $4$- and $3$-approximations, respectively. The corresponding algorithms calculate minimum cost spanning forests in certain graphs to bound the number of components in the solutions.
However, in order to satisfy the degree demands, we need to add a multigraph with a large number of edges (self-loops or a transportation problem solution), that increases the cost of the solution significantly.

To circumvent an analogous difficulty in the single-agent MVTSP, B{\'e}rczi et al.~\cite{BergerKMV2020} computed a multigraph of cost at most the optimum such that 1) the graph was connected, and 2) satisfied the degree bounds up to a small additive error.
Their approach yields a $\nicefrac32$-approximation for MVTSP after correcting the degrees by adding a matching to the multigraph. 
However, the last step does not generalize to the case of multiple agents; see the discussion at the end of the paper or Xu and Rodrigues~\cite{XuRodrigues2017} for further details.
Therefore, instead of adding a matching to fix the degrees, we use an edge-doubling approach, which yields $2$-approximations for both problem variants.

Formally, the \textsc{Minimum Bounded Degree $m$-Component Multigraph} problem consists of a complete graph $G(V, E)$ with edge costs $c(uv)$ for every $uv \in E$, degree requirements~$\rho(v)$ for each vertex $v \in V$, and an integer $m\in\mathbb{Z}_{\geq 1}$.
The objective is to find a multigraph~$X$ of minimum cost that has at most $m$ components and satisfies the degree requirements, i.e. $\ddelta_X(v)=\rho(v)$ for all $v \in V$. The following result was stated in B{\'e}rczi et al.~\cite{BercziMV2020} for $m=1$, but the theorem holds for arbitrary $m$ as multigraphs with a given number of edges and having at most $m$-components correspond to the integer points of a base-polymatroid. 

\begin{theorem}[B{\'e}rczi et al.~\cite{BercziMV2020}]
\label{thm:gpolym}
Given an instance $(G, c, \rho)$, there is an algorithm that outputs a connected multigraph $X$ of cost at most the optimum, such that each vertex $v \in V$ has degree at least $\rho(v)-1$ in $X$.
The algorithm runs in time polynomial in $n$ and $\log \sum_{v \in V} \rho(v)$.\hfill\qed
\end{theorem}

\subsection{Applying shortcuts efficiently}

  Our algorithms for the unrestricted and restricted MV-mTSP\0 problems will rely on shortcutting tours efficiently. 
  Let $X$ be a connected multigraph on vertex set $V(X)$ such that $\ddelta_X(v)$ is even for every vertex $v \in V(X)$. Furthermore, let $c$ be metric edge costs, and let $\rho(v)\leq\ddelta_X(v)$ be an even requirement for $v\in V(X)$.
  Our goal is to find a multigraph~$X'$ with cost at most $\cost(X)$ such that $\ddelta_{X'}(v) := \rho(v)$ holds for $v\in V(X)$.
  The difficulty comes from the fact that the degree surplus $\gamma(v) := \ddelta_X(v)-\rho(v)$ might not be polynomially bounded by the input size.
  \Cref{alg:shortcuts} finds a suitable multigraph $X'$ and its time complexity has a logarithmic dependence on the $\gamma$-values.
  
\begin{procedure}[!ht]
\refstepcounter{procedure}
\setcounter{procedure}{0}
  \label{alg:shortcuts}
  \caption{An efficient way of performing shortcuts}
  \begin{algorithmic}[1]
    \Statex \textbf{Input:} A connected multigraph $X$ with even degrees $\ddelta_X(v)$ and even degree requirement $\rho(v)$ for every $v \in V(X)$.
    \Statex \textbf{Output:} A connected multigraph $X'$ with degrees $\ddelta_{X'}(v) = \rho(v)$ for every $v \in V(X)$.
    \State  Use \textbf{CovertToSequence} from~\cite{GrigorievvandeKlundert2006} to calculate $\mathcal{C} = \{(C, \mu_C)\}$, a cycle decomposition of $X$.
    \State  Construct a multigraph $A$ with one copy of each cycle $C \in \mathcal{C}$.
    \State  Calculate an Eulerian trail $\eta$ in $A$.\label{st:shortcuts_trail}
    \State  Calculate an implicit Eulerian trail $\eta'$ in $X$.
    \State  Remove the last $\gamma(v) := \ddelta_X(v)-\rho(v)$ occurrences of $v$ in $\eta'$, denote the resulting multigraph by $X'$.
    \Statex  \textbf{return} $X'$
  \end{algorithmic}
\end{procedure}

\begin{lemma}
\label{lem:shortcuts}
  \Cref{alg:shortcuts} runs in time polynomial in $\abs{V(X)}$ and $\log \sum_v \gamma(v)$, where $\gamma(v) = \ddelta_X(v) - \rho(v)$.
\end{lemma}
\begin{proof}
  Finding a cycle-decomposition of a circulation is an equivalent problem to decomposing the multigraph $X$ into cycles, and it can be done efficiently.
  One implementation of this result is the algorithm \textbf{ConvertToSequence} by Grigoriev and van de Klundert~\cite{GrigorievvandeKlundert2006}.
  We use this algorithm to decompose $X$ into a collection $\mathcal{C} = \{(C, \mu_C)\}$ of $k=O(n^2)$ different cycles with multiplicities; taking $\mu_C$ copies of each cycle $C$ gives back the edges of $X$.
  The output is calculated in $O(n^4)$ steps.

  The shortcutting procedure goes similarly to the one in the proof of Theorem~1 in B{\'e}rczi et al.~\cite{BercziMV2020}.
  However, this time extra attention is needed, as the total surplus is not bounded by a polynomial of the input size.
  Let us construct a multigraph~$A$ on the vertex set $V(X)$ by taking \textit{one} copy of each cycle $C$ from $\mathcal{C}$.
  Now calculate an Eulerian trail $\eta$ in $A$; this is possible to do since the degrees in $A$ are even, moreover, this can be done in polynomial time due to the polynomial size of $A$ using the approach of Hierholzer~\cite{Fleischner1991,Hierholzer1873}. 
  One can obtain an (implicit) Eulerian trail $\eta'$ of $X$ from $\eta$ the following way.
  Start at the first vertex of~$\eta$, denoted by~$v_1$, and traverse the vertices of $\eta$ in order.
  Every time a previously unvisited vertex~$w$ is reached, traverse every cycle $C$ containing~$w$, $\mu_C$ times; in this case we say that a cycle~$C$ is \textit{rooted} at $w$.
  The Eulerian trail $\eta'$ of $X$ can be described implicitly by listing the cycles traversed in order, along with the vertices they are rooted at:
  \begin{equation}
  \label{eq:implicit_Eulerian}
  v_1 - (C_1, \mu_{C_1}) - v_2 - (C_2, \mu_{C_2}) - \dots - v_k - (C_k, \mu_{C_k}) \enspace .
  \end{equation}
  For the sake of simplicity we index the cycles and vertices based on their traverse order in the sequence.
  Note that the vertices $v_i$ might repeat multiple times, as different cycles can be rooted at the same vertex.
  Also note that the description in \Cref{eq:implicit_Eulerian} consists of $O(n^2)$ elements of size $O(n \log r(V))$, therefore it can be regarded as a compact representation of $\eta'$.
    
  Now we turn to taking shortcuts.
  Denote the visit surplus of a vertex $v$ by $\gamma(v)$, that is, $\gamma(v) := \ddelta_X(v) - \rho(v)$.
  Take a vertex $v \in V$, with $\gamma(v) > 0$, and let $\mathcal{C}^{(v)} := \{(C_1^{(v)}$, $\mu_{C_1^{(v)}})$, \dots, $(C_s^{(v)}, \mu_{C_s^{(v)}})\}$ denote the $s$ unique cycles from $\mathcal{C}$ with multiplicities, that contain the vertex $v$, in the order these cycles appear in~$\eta'$.
  Now apply shortcuts to the last $\gamma(v)$ cycles from $\mathcal{C}^{(v)}$: suppose~$s'$ is the first index, such that $\mu_{C_{s'+1}^{(v)}} + \dots + \mu_{C_{s}^{(v)}} < \gamma(v)$, i.e.\ the last $\gamma(v)$ cycles of~$\mathcal{C}^{(v)}$ consists of \textit{some} copies of cycles $C_{s'}^{(v)}$ and \textit{all} copies of cycles $C_{s'+1}^{(v)}, \dots, C_{s}^{(v)}$.
  Applying one shortcut means either (1) removing self loop $vv$, or (2) replacing 1-1 copy of the edges $uv$ and~$vw$ by $uw$ for some preceding vertex~$u$ and subsequent vertex $w$ ($u$ and $w$ might be the same vertex).
  Perform these operations in $X$, on the edges of all copies of cycles $C_{s'+1}^{(v)}, \dots, C_{s}^{(v)}$ and some copies of cycle~$C_{s'}^{(v)}$, such that the total number of shortcuts taken is $\gamma(v)$.

  Since the number of unique cycles in~$\mathcal{C}$ is $O(n^2)$, decreasing the degree of vertex $v$ by $\gamma(v)$ using operations 1) and 2) can be done by modifying $O(n^2)$ values with a total amount of $O(\log r(v))$.
  Lastly, calculate an Eulerian trail $\eta'$ on the updated multigraph $X'$ the way described above.
  Perform these steps for every vertex $v \in V$ with a positive $\gamma(v)$ value, in a total of $O(n^3 \log \gamma(v))$ steps.
  
  In the case the tours contain depots, one can drop all but one depots from a component by applying shortcuts in the cycles containing those depots, as there is a polynomial number of such cycles.
  Moreover, one can also make sure that they do not remove all depots from a component, by using a depot as a starting vertex when calculating the Eulerian trail in \Cref{st:shortcuts_trail} of the procedure.
\end{proof}

\subsection{$2$-approximation for unrestricted MV-mTSP with empty tours}

\begin{algorithm}[!ht]
  \caption{A $2$-approximation for the unrestricted MV-mTSP\0 \label{alg:mamvtsp0_2apx}}
  \begin{algorithmic}[1]
    \Statex \textbf{Input:} A complete undirected graph $G(V,E)$, costs $c:E\rightarrow\mathbbm{R}_{\geq 0}$ satisfying the triangle inequality, requests $r:V\rightarrow\mathbbm{Z}_{\geq 1}$, number of agents $m$.
    \Statex \textbf{Output:}  $m$ MVTSP tours that visit each $v \in \V$ a total of $r(v)$ times. 
    \State  Use \Cref{thm:gpolym} to obtain a multigraph $X'$ on $V$ having $r(V)$ edges and at most $m$ components, such that $\ddelta_{X'}(v) \geq 2 \cdot r(v) - 1$ for all $v \in V$.\label{st:mamvtsp0_2_kcomp}
    \State  Duplicate all edges in $X'$, denote the resulting multigraph by $X$.\label{st:mamvtsp0_2_double}
    \State  Apply shortcuts in each component of $X$ until the degree of every node $v \in V$ is~$2\cdot r(v)$, using \Cref{alg:shortcuts}.\label{st:mamvtsp0_2_shortcut}
    \Statex  \textbf{return} $X$
  \end{algorithmic}
\end{algorithm}

\begin{theorem}
\label{thm:mamvtsp0_2apx}
\Cref{alg:mamvtsp0_2apx} provides a $2$-approximation for unrestricted MV-mTSP\0 in time polynomial in $n$, $m$ and $\log r(V)$.
\end{theorem}
\begin{proof}
  The proof goes as follows.
\paragraph{Feasibility.}
  In \Cref{st:mamvtsp0_2_kcomp}, we calculate an approximate solution to the \textsc{Minimum Bounded Degree $m$-component Multigraph} problem on $V$ with $\rho(v) = 2 \cdot r(v)$ for all $v \in V$.
  The resulting multigraph $X'$ has at most $m$ components and has degree at least $2 \cdot r(v)-1$ for every $v \in V$.  
  After doubling the edges, $X$ still has at most $m$ components, and has an even degree at least $4 \cdot r(v) - 2\geq 2\cdot r(v)$ for every $v \in V$.
After taking shortcuts using \Cref{alg:shortcuts}, the number of components remains at most $m$, each vertex $v \in V$ being visited exactly $r(v)$ times in total.
    
\paragraph{Cost of solution.}
  Let us fix an optimal solution $\X$.
  The cost of $X'$ is at most $\cost(\X)$, therefore the cost of~$X$ is at most $2 \cdot \cost(\X)$.    
    Applying shortcuts in \Cref{st:mamvtsp0_2_shortcut} cannot increase the cost of $X$ as the edge costs satisfy the triangle inequality.

\paragraph{Complexity analysis.}
  Due to \Cref{thm:gpolym}, \Cref{st:mamvtsp0_2_kcomp} can be done in polynomial time.
  The graph operations in \Cref{st:mamvtsp0_2_double} can be done efficiently.
 \Cref{st:mamvtsp0_2_kcomp} constructs a multigraph with $r(V)$ edges in total, hence, after duplicating the edges in \Cref{st:mamvtsp0_2_double} of \Cref{alg:mamvtsp0_2apx}, the resulting multigraphs have a total degree surplus of $r(V)$.  
   We can use \Cref{alg:shortcuts} to the components of the multigraph calculated in \Cref{st:mamvtsp0_2_double} of \Cref{alg:mamvtsp0_2apx} to decrease the total degree of each vertex $v$ to $2 \cdot r(v)$.
  According to \Cref{lem:shortcuts}, this operation can be done in time polynomial in $\abs{V}$,~$m$ and $\log \gamma(V)$, which implies a time complexity polynomial in $n$, $m$ and $\log r(V)$.
\end{proof}

\subsection{$2$-approximation for MV-mTSP with empty tours}

\begin{algorithm}[!ht]
  \caption{A $2$-approximation for the MV-mTSP\0 \label{alg:mdmvtsp0_2apx}}
  \begin{algorithmic}[1]
    \Statex \textbf{Input:} A complete undirected graph $G(V,E)$, costs $c:E\rightarrow\mathbbm{R}_{\geq 0}$ satisfying the triangle inequality, depot set $D \subseteq V$, number of agents $m=\abs{D}$, requests $r:\V\rightarrow\mathbbm{Z}_{\geq 1}$.
    \Statex \textbf{Output:} $m$ MVTSP tours that visit each $v \in \V$ a total of $r(v)$ times.
    \State  Identify all depots $d \in D$ into one meta-depot $\hat{d}$.
            Let $\hat{V} := \V \cup \{\hat{d}\}$ and $\hat{E} := \hat{V} \times \hat{V}$.
            Set $c_{\hat{d}, v} = \min_{d \in D} c_{d, v}$ for every vertex $v \in V$, and let $d(v) := \argmin_{d \in D} c_{d, v}$. 
            Set $\hat{r}(v) = r(v)$ for $v \in \hat{V} \rep \{\hat{d}\}$.\label{st:mdmvtsp_2_metadepot}
    \For {$\mu = 1, \dots, m$}
    \State  Set $\hat{r}(d) := \mu$.
    \State  \begin{varwidth}[t]{0.9\textwidth}
    Use \Cref{thm:gpolym} to obtain a connected multigraph $\hat{X}\k$ on $\hat{V}$ having $\hat{r}(\hat{V})$ edges, such that $\ddelta_{\hat{X}\k}(v) \geq 2 \cdot \hat{r}(v) - 1$ for all $v \in \hat{V}$.\label{st:mdmvtsp_2_kcomp}
    \end{varwidth}
    \State  Duplicate all edges in $\hat{X}\k$.\label{st:mdmvtsp_2_double}
    \State  \begin{varwidth}[t]{0.9\textwidth}
    Replace $\hat{d}$ with $d_1, \dots, d_m$ and all copies of edges $(\hat{d}, v)$ with $(d(v), v)$ for $v \in V$ in $\hat{X}\k$; denote the resulting multigraph by $X\k$.\label{st:mdmvtsp_2_expand}
    \end{varwidth}
    \State  \begin{varwidth}[t]{0.9\textwidth}
    If there is a component in $X\k$ with more than one depot $d \in D$, disconnect depots until there is only one depot per component, using shortcuts.\label{st:mdmvtsp_2_shortcut1}
    \end{varwidth}
    \State \begin{varwidth}[t]{0.9\textwidth}
    Apply \Cref{alg:shortcuts} in each non-trivial component of $X\k$ until the degree of every vertex $v \in \V $ becomes $r(v)$.\label{st:mdmvtsp_2_shortcut2}
    \end{varwidth}
    \EndFor
  \Statex  \textbf{return} the cheapest $X\k$.
  \end{algorithmic}
\end{algorithm}

\begin{theorem}
\label{thm:mdmvtsp0_2apx}
\Cref{alg:mdmvtsp0_2apx} provides a $2$-approximation for MV-mTSP\0 in time polynomial in $n$, $m$ and $\log r(V)$.
\end{theorem}
\begin{proof}
  The proof goes as follows.
\paragraph{Feasibility.}
  For every $\mu = 1, \dots, m$, calculate an approximate solution to the \textsc{Minimum Bounded Degree $1$-component Multigraph} problem on $\hat{V}$ with $\rho(v) = 2 \cdot \hat{r}(v)$ for all $v \in \hat{V}$.
  The resulting multigraph $\hat{X}\k$ is connected and has degree at least $2 \cdot \hat{r}(v)-1$ for every $v \in \hat{V}$.  
  After doubling the edges, $\hat{X}\k$ remains connected, and every vertex $v\in \V$ has an even degree at least $4\cdot \hat{r}(v)-2\geq 2\cdot r(v)$.
  By replacing~$\hat{d}$ with $d_1, \dots, d_\mu$, the resulting multigraph has at most $\mu$ components. After disconnecting depots and taking shortcuts, the number of components will be exactly $m$ with exactly one depot in each component.
  Furthermore, each vertex $v\in V$ is visited exactly $r(v)$ times in total.
    
\paragraph{Cost of solution.}
  Let us denote by $m^\star$ the number of components in the optimal solution~$\X$, where $1\leq m^\star\leq m$. 
  The multigraph~$\hat{X}\k$ obtained in \Cref{st:mdmvtsp_2_kcomp} for $\mu=m^\star$ has cost at most $\cost(X^\star)$, therefore after doubling all its edges in \Cref{st:mdmvtsp_2_double} the resulting multigraph $X\k$ has cost at most $2 \cdot \cost(X^\star)$.
  \Cref{st:mdmvtsp_2_expand} does not change the cost of $X\k$ as each edge $(d', v)$ is replaced with an edge $(d(v), v)$ of identical cost.
  Finally, the disconnections and shortcuts in \Cref{st:mdmvtsp_2_shortcut1,st:mdmvtsp_2_shortcut2} cannot increase the cost of $X\k$ as the costs satisfy the triangle inequality.
    
\paragraph{Complexity analysis.}
Due to \Cref{thm:gpolym}, \Cref{st:mdmvtsp_2_kcomp} can be done in polynomial time.
  The graph operations in \Cref{st:mdmvtsp_2_metadepot,st:mdmvtsp_2_double,st:mdmvtsp_2_expand} can be done efficiently.
 \Cref{st:mdmvtsp_2_kcomp} constructs a multigraph with $\hat{r}(\hat{V})\leq r(V)+2\cdot m$ edges in total, hence, after duplicating the edges in \Cref{st:mdmvtsp_2_double} of \Cref{alg:mdmvtsp0_2apx}, the resulting multigraph have a total degree surplus of $r(V)$.  
   We can use \Cref{alg:shortcuts} to the components of the multigraph calculated in \Cref{st:mdmvtsp_2_double} of \Cref{alg:mdmvtsp0_2apx} to decrease the total degree of each vertex $v$ to $2 \cdot r(v)$.
  According to \Cref{lem:shortcuts}, this operation can be done in time polynomial in $\abs{V}$,~$m$ and $\log \gamma(V)$, which implies a time complexity polynomial in $n$, $m$ and $\log r(V)$.
\end{proof}

\section{Discussion and open problems}
In this paper, we introduced different generalizations of the multiple TSP by considering many visits (possibly exponentially many), and provided polynomial-time constant approximation algorithms for each of them.
We started by $4$- and $3$-approximations for the unrestricted and restricted variants with arbitrary tours, respectively. 
The proofs rely on the idea of first building a constrained spanning forest, then adding self-loops in the unrestricted setting or adding a transportation problem solution in the restricted setting.
We then provided $4$-approximations for the disjoint tours variants, relying again on the addition of self-loops.
Finally we gave $2$-approximation algorithms for the empty tours variants, that calculate a constrained multigraph using the algorithm from~\cite{BercziMV2020}, followed by an edge-doubling and an efficient shortcutting procedure.

Several interesting open questions remain.

\paragraph{Improving the constant factors.}
Given the $\nicefrac32$-approximation to the TSP by Christofides~\cite{Christofides1976} and Serdyukov~\cite{Serdyukov1978}, and the recent improvement by Karlin et al.~\cite{KarlinKG2021}, a major open problem is to give an approximation algorithm for any of the multiple-agent or multiple-depot problem variants with a ratio strictly better than $2$.
The MV-mTSP with empty tours is a more general problem than the multidepot mTSP considered by Rathinam et al.~\cite{RathinamSD2007} and Xu et al.~\cite{XuXR2011}.
By \Cref{lem:sumc_reduction}, any improvement to the MV-mTSP with arbitrary tours or disjoint tours would imply an improved algorithm for the MV-mTSP with empty tours.
This means that a better than $2$-approximation for any MV-mTSP variants would translate to a better than $2$-approximation for the multidepot mTSP, which would be a major breakthrough.

Another open problem is whether we can generalize the $\nicefrac32$-approximation by Xu and Rodrigues~\cite{XuRodrigues2017}, that runs in polynomial time for fixed $m$, to the many-visits setting.
Adapting the g-polymatroidal approach from B{\'e}rczi et al.~\cite{BercziMV2020} (in a similar fashion to \Cref{alg:mdmvtsp0_2apx}) is not straightforward.
The main difficulty lies in generalizing the edge exchange subroutine by Xu and Rodrigues~\cite{XuRodrigues2017} from the single-visit case to the many-visits case.
In that routine, an edge exchange step removes an edge from a spanning forest $F$ and adds another edge such that the resulting graph will be a spanning forest as well.
In case of spanning forests (which are used for the single-visits case), a degree of $1$ for a vertex is sufficient, as adding a matching contributes enough degree to obtain a feasible tour.
In case of multigraphs, however, repeated edge exchange steps can decrease the degree of a vertex so much that a matching would not be sufficient.

\paragraph{Path variants.}
  Both the path TSP~\cite{TraubVygen2019,Zenklusen2019,TraubVZ2020} and the many-visits path TSP~\cite{BercziMV2020} have received reasonable attention recently, with the best polynomial-time approximation ratio being~$\nicefrac32$.
  The multiple agent variants, referred to as multiple-depot multiple-terminal Hamiltonian path problem admit a $2$-approximations~\cite{RathinamSengupta2006,BaeRathinam2012}, a $(2-\frac1m)$-approximation when the terminals are not fixed~\cite{YangLY2021}, and a $(2-\frac{1}{2m+1})$-approximation~\cite{YangLiu2019} when empty tours are allowed.
  An interesting direction of future research could be characterising the path versions of the different many-visits mTSP variants.
  
  The $4$\hyp{}approximations provided in \Cref{alg:mamvtsp_4apx,alg:mdmvtsp_4apx} seems to be simple enough to be adapted to $4$-approximations for the path versions of many-visits mTSP.
  These modifications would be using Hamiltonian paths instead of cycles in \Cref{alg:mamvtsp_4apx} and finding a constrained spanning forest with exactly one depot and one terminal in each tree using a matroid intersection algorithm~\cite{BaeRathinam2012,RathinamSengupta2006} in \Cref{alg:mdmvtsp_4apx}.
  However, for the other problem variants, adapting the existing algorithms is not straightforward.
  The arguments in \Cref{clm:three_four,clm:seven_eight} do not carry over for the path case, so approximating the path versions of problems P3--P4 and P7--P8 might require approaches completely different from the ones presented in this paper.

\paragraph{Min-max variants.}
  Throughout the paper, the objective was to minimize the total cost of the tours. However, it would be equally natural to consider minimizing the longest (most expensive) tour of an agent.
  The single-visit counterparts of such problems have been referred to as Min-Max Cycle Cover and Min-Max Routing Problems.
  The first constant-factor approximation result is due to Frederickson et al.~\cite{FredericksonHK1978}, and many other followed considering problems without depots~\cite{KhaniSalavatipour2014,XuLL2015,ArkinHL2006,YuLiu2016}, as well as with depots with different restrictions on the tours~\cite{XuXZ2012,XuLL2015,EvenGKRS2004,Jorati2013,YuLiu2016}.
  The approximation guarantees of these approaches range from $5/2$ to $8+\eps$.
  
  The many-visits counterparts of these problems also have connections to scheduling theory, in particular, they can be regarded as high-multiplicity scheduling problems with sequence-depending setup times, where the objective is to minimize the makespan of the schedule, commonly denoted by $C_{\max}$.

{\small
\paragraph{Acknowledgements.}
This research was supported by DAAD with funds of the Bundesministerium f{\"u}r Bildung und Forschung (BMBF).
Krist\'of B\'erczi was supported by the J\'anos Bolyai Research Fellowship and the Lend\"ulet Programme of the Hungarian Academy of Sciences -- grant number LP2021-1/2021, by the Hungarian National Research, Development and Innovation Office -- NKFIH, grant number FK128673, and by the Thematic Excellence Programme -- TKP2020-NKA-06 (National Challenges Subprogramme).
Roland Vincze was partially supported by DFG grant MO 2889/3-1 and MN 59/4-1.}

\bibliographystyle{abbrvnat}
\bibliography{bibliography} 

\end{document}